\documentclass{article}
\usepackage[utf8x]{inputenc}
\usepackage[a4paper,textwidth=15cm,textheight=22.5cm]{geometry}

\title{Chains, Koch Chains, and Point Sets with many Triangulations}

\date{February 24, 2023}

\author{
  Daniel Rutschmann\footnote{Department of Computer Science, ETH Z\"urich, Switzerland. E-mail: \url{daru@dtu.dk}}
    \and
  Manuel Wettstein\footnote{Department of Computer Science, ETH Z\"urich, Switzerland. E-mail: \url{manuelwe@inf.ethz.ch}}
}

\usepackage{amsmath}
\usepackage{amsfonts}
\usepackage{amsthm}
\usepackage{mathtools}
\usepackage{hyperref}
\usepackage[capitalize]{cleveref}
\usepackage{booktabs}
\usepackage{tikz}

\newcommand{\bR}{\mathbb{R}}
\newcommand{\bZ}{\mathbb{Z}}

\newcommand{\ppa}[1]{\left(#1\right)}
\DeclarePairedDelimiter\bra{[}{]}

\newcommand{\eps}{\ensuremath\varepsilon}

\newcommand{\cOne}{E}

\newcommand{\flip}[1]{\overline{#1}}
\newcommand{\vex}[1]{C_{\mathsf{cvx}}(#1)}
\newcommand{\cave}[1]{C_{\mathsf{ccv}}(#1)}
\newcommand{\cDC}[1]{C_{\mathsf{dbl}}(#1)}
\newcommand{\cZZ}[1]{C_{\mathsf{zz}}(#1)}
\newcommand{\cDZZ}[1]{C_{\mathsf{dzz}}(#1)}
\newcommand{\koch}[1]{K_{#1}}

\newcommand{\SC}{\mathrm{SC}}

\newcommand{\Ctwin}{C_{\mathsf{twin}}}
\newcommand{\Cpoly}{C_{\mathsf{poly}}}
\newcommand{\Cgdc}{C_{\mathsf{gdc}}}

\newcommand{\tOmega}{\widetilde{\Omega}}

\newcommand{\tTheta}{\widetilde{\Theta}}

\newcommand{\DP}{\mathrm{DP}}
\DeclareMathOperator{\TR}{tr}
\DeclareMathOperator{\Tr}{tr}

\usetikzlibrary{lindenmayersystems}
\tikzset{
  Koch curve/.style = {
    l-system={
      rule set={F -> F++F----F++F},
      axiom=-F++F-,
      step=1pt,
      angle=30,
      #1
    }
  }
}

\definecolor{myred}{RGB}{231,76,60}
\definecolor{myblue}{RGB}{41,128,185}

\newtheorem{definition}{Definition}
\crefname{definition}{Definition}{Definitions}
\Crefname{definition}{Definition}{Definitions}

\newtheorem{proposition}[definition]{Proposition}
\crefname{proposition}{Proposition}{Propositions}
\Crefname{proposition}{Proposition}{Propositions}

\newtheorem{lemma}[definition]{Lemma}
\crefname{lemma}{Lemma}{Lemmas}
\Crefname{lemma}{Lemma}{Lemmas}

\newtheorem{theorem}[definition]{Theorem}
\crefname{theorem}{Theorem}{Theorems}
\Crefname{theorem}{Theorem}{Theorems}

\newtheorem{corollary}[definition]{Corollary}
\crefname{corollary}{Corollary}{Corollaries}
\Crefname{corollary}{Corollary}{Corollaries}

\newtheorem{example}[definition]{Example}
\crefname{example}{Example}{Examples}
\Crefname{example}{Example}{Examples}

\crefname{figure}{Figure}{Figures}
\Crefname{figure}{Figure}{Figures}

\crefname{table}{Table}{Tables}
\Crefname{table}{Table}{Tables}
\usetikzlibrary{calc,patterns,snakes,backgrounds}
\tikzset{
  point/.style = {
    circle,
    draw=black,
    fill=black,
    inner sep=0.35mm
  },
  edge/.style = {
    draw=black!75!white
  },
  unavoidable/.style = {
    draw=gray
  },
  chain/.style = {
    draw=black,
    thick
  },
  highlight/.style = {
    myblue,
    opacity=0.8
  }
}

\begin{document}

\maketitle

\begin{abstract}
  We introduce the abstract notion of a chain, which is a sequence of $n$ points in the plane, ordered by $x$-coordinates, so that the edge between any two consecutive points is unavoidable as far as triangulations are concerned.
  A general theory of the structural properties of chains is developed, alongside a general understanding of their number of triangulations.
  
  We also describe an intriguing new and concrete configuration, which we call the Koch chain due to its similarities to the Koch curve.
  A specific construction based on Koch chains is then shown to have $\Omega(9.08^n)$ triangulations.
  This is a significant improvement over the previous and long-standing lower bound of $\Omega(8.65^n)$ for the maximum number of triangulations of planar point sets.
\end{abstract}

\section{Introduction}

Let $P$ be a set of $n$ points in the Euclidean plane.
Throughout the paper, $P$ is assumed to be in \emph{general position}, which means for us that no two points have the same $x$-coordinate and that no three points are on a common line.
A \emph{geometric graph} on $P$ is a graph with vertex set $P$ combined with an embedding into the plane where edges are realized as straight-line segments between the corresponding endpoints.
It is called \emph{crossing-free} if the edges have no pairwise intersection, except possibly in a common endpoint.

\subparagraph{Triangulations.}

\def\tr{\mathrm{tr}}
\def\trmax{\tr_{\mathrm{max}}}
\def\trmin{\tr_{\mathrm{min}}}

Perhaps the most prominent and most studied family of crossing-free geometric graphs is the family of \emph{triangulations}, which may be defined simply as edge-maximal crossing-free geometric graphs on $P$.
It is easy to see that such a definition implies that the edges of any triangulation subdivide the convex hull of $P$ into triangular regions.

Let $\tr(P)$ denote the number of triangulations on a given point set $P$.
Trying to better understand this quantity is a fundamental question in combinatorial and computational geometry.
For very specific families of point sets, exact formulas or at least asymptotic estimates can be derived.
For example, it is well-known that if $P$ is a set of $n$ points in \emph{convex position}, then~$\tr(P) = C_{n-2}$, where $C_k = \frac{1}{k+1}\binom{2k}{k} = \Theta(k^{-3/2}4^k)$ is the $k$-th Catalan number~\cite{oeiscatalan}.
In general, however, this problem turns out to be much more elusive.

There is an elegant algorithm by Alvarez and Seidel~\cite{alvarez2013} from 2013 that computes $\tr(P)$ in exponential time $O(2^nn^2)$.
It has been surpassed by  Marx and Miltzow~\cite{marx2016} in 2016, who showed how to compute $\tr(P)$ in subexponential time $n^{O(\sqrt n)}$.
Moreover, Avis and Fukuda~\cite{avis1996} have shown already in 1996 how to enumerate the set of all triangulations on $P$ (i.e., to compute an explicit representation of each element) by using a general technique called reverse search in time~$\tr(P) \cdot p(n)$ for some polynomial $p$.
A particularly efficient implementation of that technique with~$p(n) = O(\log \log n)$ has been described by Bespamyatnikh~\cite{bespamyatnikh2002}.

Extensive research has also gone into extremal upper and lower bounds in terms of the number of points.
That is, if we define
\begin{align*}
  \trmax(n) = \max_{P \colon |P|=n} \tr(P), && \trmin(n) = \min_{P \colon |P|=n} \tr(P)
\end{align*}
to be the respectively largest and smallest numbers of triangulations attainable by a set~$P$ of~$n$ points in general position, then various authors have attempted to establish and improve upper and lower bounds on these quantities.

As far as the maximum is concerned, a seminal result by Ajtai, Chv\'atal, Newborn, and Szemer\'edi~\cite{ajtai1982} from 1982 shows that the number of triangulations---and, more generally, the number of all crossing-free geometric graphs---is at most $10^{13n}$.
A long series of successive improvements~\cite{smith1989,denny1997,seidel1998,santos2003,sharir2006} using a variety of different techniques has culminated in the currently best upper bound $\trmax(n) \leq 30^n$ due to Sharir and Sheffer~\cite{sharir2011}, which has remained uncontested for over a decade.
Coming from the other side, attempts have been made to construct point sets with a particularly large number of triangulations.
For some time, the \emph{double chain} by Garc\'ia, Noy, and Tejel~\cite{garcia2000} with approximately $\Theta(8^n)$ triangulations was conjectured to have the largest possible number of triangulations.
However, variants like the \emph{double zig-zag chain} by Aichholzer et al.~\cite{aichholzer2007} with $\Theta(8.48^n)$ triangulations and a specific instance of the \emph{generalized double zig-zag chain} by Dumitrescu, Schulz, Sheffer, and T\'oth~\cite{dumitrescu2013} with $\Omega(8.65^n)$ triangulations have since been discovered.
But also on this front, no further progress on the lower bound $\trmax(n) = \Omega(8.65^n)$ has been made for a decade.

The situation for the minimum is different insofar that the \emph{double circle} with $\Theta(3.47^n)$ triangulations, as analyzed by Hurtado and Noy~\cite{hurtado1997} in 1997, is still conjectured by many to have the smallest number of triangulations.
In other words, it is believed that the resulting upper bound $\trmin(n) = O(3.47^n)$ is best possible.
On the other hand, Aichholzer et al.~\cite{aichholzer2016} have shown that every point set has at least $\Omega(2.63^n)$ triangulations, thereby establishing the lower bound $\trmin(n) = \Omega(2.63^n)$.

The focus of this paper lies on $\trmax(n)$ and, more specifically, on establishing an improved lower bound on that quantity.
Ultimately, we show how to construct a new infinite family of point sets with $\Omega(9.08^n)$ triangulations, thereby proving $\trmax(n) = \Omega(9.08^n)$.

\subparagraph{General chains.}

It has occurred to us that almost all families of point sets whose numbers of triangulations have been analyzed over the years have a very special structure, which we are trying to capture in the following definition.

\begin{definition}
  \label{def:chain}
  A \emph{chain} $C$ is a sequence of points $p_0,\dots,p_n$ sorted by increasing $x$-coordinates, such that the edge $p_{i-1}p_i$ is \emph{unavoidable} (i.e., contained in every triangulation of $C$) for each~$i=1,\dots,n$.
  These specific unavoidable edges are also referred to as \emph{chain edges}.
\end{definition}

In contrast to previous convention, we use the parameter $n$ to denote the number of chain edges and not the number of points in $C$, which is $n+1$.
Also note that \cref{def:chain} implies that the edge $p_0p_n$ is an edge of the convex hull and, hence, also unavoidable.
Indeed, since all chain edges are unavoidable, the edge $p_0p_n$ cannot possibly cross any of them and, hence, is either above or below all the points in between.
Therefore, a chain always admits a spanning cycle of unavoidable edges with at least one hull edge.
We prove in \cref{sec:structure} that this is also a characterization of chains in terms of \emph{order types} (see \cite{goodman1983} for a definition).

\begin{theorem}
  \label{thm:characterization}
  For every point set that admits a spanning cycle of unavoidable edges including at least one convex hull edge, there exists a chain with the same order type.
\end{theorem}

All of the mentioned families of point sets (convex position, double chain, and so on) are usually neither defined nor depicted in a way that makes it clear that they may be thought of as chains as in \cref{def:chain}.
Still, the premise of \cref{thm:characterization} is easily verified for all of them except for the double circle, which may however be transformed into a chain by removing one of the inner points.
\Cref{fig:examples} shows realizations of some such point sets as chains.

\begin{figure}
  \newcommand{\totalY}{2.5/2}
  \centering
  \begin{tikzpicture}
    \def\xshift{90}
    \def\labelshift{1.8}
    \begin{scope}[xshift=0*\xshift]
      \newcommand{\rad}{\totalY}
      \newcommand{\ang}{360/10}
      \node[point] (p0) at (-0*\ang:\rad) {};
      \node[point] (p1) at (-0.5*\ang:\rad) {};
      \node[point] (p2) at (-1*\ang:\rad) {};
      \node[point] (p3) at (-1.5*\ang:\rad) {};
      \node[point] (p4) at (-2*\ang:\rad) {};
      \node[point] (p5) at (-2.5*\ang:\rad) {};
      \node[point] (p6) at (-3*\ang:\rad) {};
      \node[point] (p7) at (-3.5*\ang:\rad) {};
      \node[point] (p8) at (-4*\ang:\rad) {};
      \node[point] (p9) at (-4.5*\ang:\rad) {};
      \node[point] (p10) at (-5*\ang:\rad) {};
      \begin{scope}[on background layer]
        \draw[unavoidable] (p10) -- (p0);
        \draw[chain] (p0) -- (p1) -- (p2) -- (p3) -- (p4) -- (p5)
          -- (p6) -- (p7) -- (p8) -- (p9) -- (p10);
      \end{scope}
      \node at (0,-\labelshift) {\footnotesize convex position};
    \end{scope}%
    \begin{scope}[xshift=1*\xshift]
      \newcommand{\rad}{\totalY}
      \newcommand{\radd}{0.9*\rad}
      \newcommand{\ang}{360/10}
      \node[point] (p0) at (-0*\ang:\rad) {};
      \node[point] (p1) at (-1*\ang:\rad) {};
      \node[point] (p2) at (-2*\ang:\rad) {};
      \node[point] (p3) at (-3*\ang:\rad) {};
      \node[point] (p4) at (-4*\ang:\rad) {};
      \node[point] (p5) at (-5*\ang:\rad) {};
      \node[point] (q0) at (-0.5*\ang:\radd) {};
      \node[point] (q1) at (-1.5*\ang:\radd) {};
      \node[point] (q2) at (-2.5*\ang:\radd) {};
      \node[point] (q3) at (-3.5*\ang:\radd) {};
      \node[point] (q4) at (-4.5*\ang:\radd) {};
      \begin{scope}[on background layer]
        \draw[unavoidable] (p0) -- (p1) -- (p2) -- (p3) -- (p4) -- (p5) -- (p0);
        \draw[chain] (p0) -- (q0) -- (p1) -- (q1) -- (p2) -- (q2)
          -- (p3) -- (q3) -- (p4) -- (q4) -- (p5);
      \end{scope}
      \node at (0,-\labelshift) {\footnotesize \phantom{p}``double circle''\phantom{p}};
    \end{scope}%
    \begin{scope}[xshift=2*\xshift]
      \newcommand{\totalX}{3.5*\totalY}
      \newcommand{\rad}{\totalY}
      \newcommand{\ang}{360/6/4}
      \node[point] (p0) at (-\rad,0) {};
      \node[point] (p4) at (-\rad/2,-\rad) {};
      \node[point] (p1) at ($(p0)!0.25!(p4)+0.05*(2,1)$) {};
      \node[point] (p2) at ($(p0)!0.5!(p4)+0.07*(2,1)$) {};
      \node[point] (p3) at ($(p0)!0.75!(p4)+0.05*(2,1)$) {};
      \node[point] (q0) at (\rad,0) {};
      \node[point] (q4) at (\rad/2,-\rad) {};
      \node[point] (q1) at ($(q0)!0.25!(q4)+0.05*(-2,1)$) {};
      \node[point] (q2) at ($(q0)!0.5!(q4)+0.07*(-2,1)$) {};
      \node[point] (q3) at ($(q0)!0.75!(q4)+0.05*(-2,1)$) {};
      \begin{scope}[on background layer]
        \draw[unavoidable] (p4) -- (p0) -- (q0) -- (q4);
        \draw[chain] (p0) -- (p1) -- (p2) -- (p3) -- (p4) -- (q4) -- (q3) -- (q2) -- (q1) -- (q0);
      \end{scope}
      \node at (0,-\labelshift) {\footnotesize \phantom{p}double chain\phantom{p}};
    \end{scope}%
    \begin{scope}[xshift=3*\xshift]
      \newcommand{\totalX}{3.5*\totalY}
      \newcommand{\rad}{\totalY}
      \newcommand{\ang}{360/6/4}
      \node[point] (p0) at (-\rad,0) {};
      \node[point] (p4) at (-\rad/2,-\rad) {};
      \node[point] (p1) at ($(p0)!0.25!(p4)+0.025*(2,1)$) {};
      \node[point] (p2) at ($(p0)!0.5!(p4)+0.07*(2,1)$) {};
      \node[point] (p3) at ($(p0)!0.75!(p4)+0.025*(2,1)$) {};
      \node[point] (q0) at (\rad,0) {};
      \node[point] (q4) at (\rad/2,-\rad) {};
      \node[point] (q1) at ($(q0)!0.25!(q4)+0.025*(-2,1)$) {};
      \node[point] (q2) at ($(q0)!0.5!(q4)+0.07*(-2,1)$) {};
      \node[point] (q3) at ($(q0)!0.75!(q4)+0.025*(-2,1)$) {};
      \begin{scope}[on background layer]
        \draw[unavoidable] (p4) -- (p0) -- (q0) -- (q4);
        \draw[chain] (p0) -- (p1) -- (p2) -- (p3) -- (p4) -- (q4) -- (q3) -- (q2) -- (q1) -- (q0);
      \end{scope}
      \node at (0,-\labelshift) {\footnotesize double zig-zag chain};
    \end{scope}%
  \end{tikzpicture}
  \caption{
    Some classic point sets realized as chains. For the double circle, we need to remove one of the inner points. Chain edges are displayed black and bold, other unavoidable hull edges in gray.
  }
  \label{fig:examples}
\end{figure}
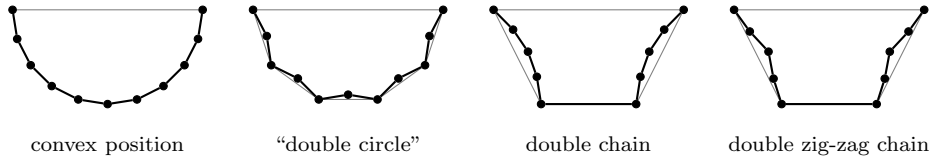

Imagine walking along the chain edges and recording at each point the information whether we make a left turn or a right turn.
It can be noted already now that such information---while crucial---is not enough to really capture all of the relevant combinatorial structure of a given chain.
Instead, the right way of looking at it turns out to be recording for each edge $p_ip_j$ whether it lies above or below all the chain edges in between.

The simple linear structure inherent to chains allows us to develop a combinatorial theory in \cref{sec:structure}, by which every chain admits a unique construction starting from the primitive chain with only one edge.
Two types of sum operations, so-called convex and concave sums, are used to ``concatenate'' chains, while an inversion allows to ``flip'' a chain on its head.
This yields for every chain a concise and unique description as an algebraic formula.
Based on this, we will also see that the number of combinatorially different chains with $n$ chain edges is equal to~$S_{n-1}$, where~$S_k = \sum_{i=0}^k \frac{1}{i+1} \binom{k}{i} \binom{k+i}{i} = \Theta(k^{-3/2}(3+\sqrt{8})^k)$ is the $k$-th large Schr\"oder number~\cite{oeisbigschroeder}.

\subparagraph{Triangulations of chains.}

The unavoidable chain edges separate every triangulation cleanly into an \emph{upper triangulation} of the region above the chain edges and into a \emph{lower triangulation} of the region below.
Therefore, both upper and lower triangulations may be analyzed separately.
It also follows that there is no further complication due to inner vertices as one would typically encounter them in general point sets.

There is a simple cubic time dynamic programming algorithm for counting triangulations of simple polygons~\cite{epstein1994}.
Such an algorithm can of course also be used to count both the upper and lower triangulations of a given chain.
However, we show in \cref{sec:triangulations} that the additional structure of chains allows us to devise an improved quadratic time algorithm, which plays a crucial role in the derivation of our main result.

\begin{theorem}
  \label{thm:algorithm}
  Given a chain $C$ with $n$ chain edges as input, it is possible to compute the number $\tr(C)$ by using only $O(n^2)$ integer additions and multiplications.
\end{theorem}

\subparagraph{The Koch chain.}

There is a particular type of chain that has caught our interest and which, to the best of our knowledge, has not been described in the literature before.
We call it the \emph{Koch chain} due to its striking similarity in appearance and definition to the famous Koch curve.
More precise definitions follow later in \cref{def:kochchain}; for now, suppose $K_0$ is a primitive chain with just one chain edge, and let the $s$-th iteration $K_s$ of the Koch chain be defined by concatenating two flipped and sufficiently flattened copies of $K_{s-1}$ in such a way that the chain edges at the point of concatenation form a left turn, see \cref{fig:kochchain}.

\begin{figure}
  \centering
  \begin{tikzpicture}[scale=2]
    \def\coords{
      \def\eps{0.04}
      \def\epss{0.01}
      \def\epsss{0.0025}
      \coordinate (c0)  at (-1,0);
      \coordinate (c16) at (1,0);
      \coordinate (c8)  at (0,-1);
      \coordinate (c4)  at ($(c0)!0.5!(c8)+(\eps,\eps)$);
      \coordinate (c12) at ($(c8)!0.5!(c16)+(-\eps,\eps)$);
      \coordinate (c2)  at ($(c0)!0.5!(c4)+(-\epss,-\epss)$);
      \coordinate (c6)  at ($(c4)!0.5!(c8)+(-\epss,-\epss)$);
      \coordinate (c10) at ($(c8)!0.5!(c12)+(\epss,-\epss)$);
      \coordinate (c14) at ($(c12)!0.5!(c16)+(\epss,-\epss)$);
      \coordinate (c1)  at ($(c0)!0.5!(c2)+(\epsss,\epsss)$);
      \coordinate (c3)  at ($(c2)!0.5!(c4)+(\epsss,\epsss)$);
      \coordinate (c5)  at ($(c4)!0.5!(c6)+(\epsss,\epsss)$);
      \coordinate (c7)  at ($(c6)!0.5!(c8)+(\epsss,\epsss)$);
      \coordinate (c9)  at ($(c8)!0.5!(c10)+(-\epsss,\epsss)$);
      \coordinate (c11) at ($(c10)!0.5!(c12)+(-\epsss,\epsss)$);
      \coordinate (c13) at ($(c12)!0.5!(c14)+(-\epsss,\epsss)$);
      \coordinate (c15) at ($(c14)!0.5!(c16)+(-\epsss,\epsss)$);
    }
    \def\coordss{
      \pgfmathsetmacro\ninthsqrtthree{sqrt(3)/9}
      \coordinate (c0) at (-1,0);
      \coordinate (c16) at (1,0);
      \coordinate (c8) at (0,-3*\ninthsqrtthree);
      \coordinate (c4) at (-1/3,0);
      \coordinate (c12) at (1/3,0);
      \coordinate (c2) at (-2/3,-1*\ninthsqrtthree);
      \coordinate (c6) at (-1/3,-2*\ninthsqrtthree);
      \coordinate (c10) at (1/3,-2*\ninthsqrtthree);
      \coordinate (c14) at (2/3,-1*\ninthsqrtthree);
      \coordinate (c1) at (-7/9,0);
      \coordinate (c3) at (-5/9,0);
      \coordinate (c5) at (-2/9,-1*\ninthsqrtthree);
      \coordinate (c7) at (-1/9,-2*\ninthsqrtthree);
      \coordinate (c9) at (1/9,-2*\ninthsqrtthree);
      \coordinate (c11) at (2/9,-1*\ninthsqrtthree);
      \coordinate (c13) at (5/9,0);
      \coordinate (c15) at (7/9,0);
    }
    \def\xshift{90}
    \def\yshift{40}
    \tikzset{every node/.style={label distance=0.5cm}}
    \begin{scope}[yshift=-0.25*\yshift]
      \begin{scope}[xshift=0*\xshift]
        \coords
        \node[point,label=left:$K_0$] (p0) at (c0) {};
        \node[point] (p1) at (c16) {};
        \begin{scope}[on background layer]
          \draw[chain] (p0) -- (p1);
        \end{scope}
      \end{scope}
      \begin{scope}[xshift=1*\xshift]
        \coordss
        \node[label=right:$\phantom{K_0}$] at (c16) {};
        \draw[chain] (c0) -- (c16);
      \end{scope}
    \end{scope}
    \begin{scope}[yshift=-1*\yshift]
      \begin{scope}[xshift=0*\xshift]
        \coords
        \node[point,label=left:$K_1$] (p0) at (c0) {};
        \node[point] (p1) at (c8) {};
        \node[point] (p2) at (c16) {};
        \begin{scope}[on background layer]
          \draw[unavoidable] (p0) -- (p2);
          \draw[chain] (p0) -- (p1) -- (p2);
        \end{scope}
      \end{scope}
      \begin{scope}[xshift=1*\xshift]
        \coordss
        \draw[chain] (c0) -- (c8) -- (c16);
      \end{scope}
    \end{scope}
    \begin{scope}[yshift=-2*\yshift]
      \begin{scope}[xshift=0*\xshift]
        \coords
        \node[point,label=left:$K_2$] (p0) at (c0) {};
        \node[point] (p1) at (c4) {};
        \node[point] (p2) at (c8) {};
        \node[point] (p3) at (c12) {};
        \node[point] (p4) at (c16) {};
        \begin{scope}[on background layer]
          \draw[unavoidable] (p0) -- (p4);
          \draw[unavoidable] (p0) -- (p2) -- (p4);
          \draw[chain] (p0) -- (p1) -- (p2) -- (p3) -- (p4);
        \end{scope}
      \end{scope}
      \begin{scope}[xshift=1*\xshift]
        \coordss
        \draw[chain] (c0) -- (c4) -- (c8) -- (c12) -- (c16);
      \end{scope}
    \end{scope}
    \begin{scope}[yshift=-3*\yshift]
      \begin{scope}[xshift=0*\xshift]
        \coords
        \node[point,label=left:$K_3$] (p0) at (c0) {};
        \node[point] (p1) at (c2) {};
        \node[point] (p2) at (c4) {};
        \node[point] (p3) at (c6) {};
        \node[point] (p4) at (c8) {};
        \node[point] (p5) at (c10) {};
        \node[point] (p6) at (c12) {};
        \node[point] (p7) at (c14) {};
        \node[point] (p8) at (c16) {};
        \begin{scope}[on background layer]
          \draw[unavoidable] (p0) -- (p8);
          \draw[unavoidable] (p0) -- (p4) -- (p8);
          \draw[chain] (p0) -- (p1) -- (p2) -- (p3) -- (p4)
            -- (p5) -- (p6) -- (p7) -- (p8);
        \end{scope}
      \end{scope}
      \begin{scope}[xshift=1*\xshift]
        \coordss
        \draw[chain] (c0) -- (c2) -- (c4) -- (c6) --(c8)
          -- (c10) -- (c12) -- (c14) -- (c16);
      \end{scope}
    \end{scope}
    \begin{scope}[yshift=-4*\yshift]
      \begin{scope}[xshift=0*\xshift]
        \coords
        \node[point,label=left:$K_4$] (p0) at (c0) {};
        \node[point] (p1) at (c1) {};
        \node[point] (p2) at (c2) {};
        \node[point] (p3) at (c3) {};
        \node[point] (p4) at (c4) {};
        \node[point] (p5) at (c5) {};
        \node[point] (p6) at (c6) {};
        \node[point] (p7) at (c7) {};
        \node[point] (p8) at (c8) {};
        \node[point] (p9) at (c9) {};
        \node[point] (p10) at (c10) {};
        \node[point] (p11) at (c11) {};
        \node[point] (p12) at (c12) {};
        \node[point] (p13) at (c13) {};
        \node[point] (p14) at (c14) {};
        \node[point] (p15) at (c15) {};
        \node[point] (p16) at (c16) {};
        \begin{scope}[on background layer]
          \draw[unavoidable] (p0) -- (p16);
          \draw[unavoidable] (p0) -- (p8) -- (p16);
          \draw[chain] (p0) -- (p1) -- (p2) -- (p3) -- (p4)
            -- (p5) -- (p6) -- (p7) -- (p8) -- (p9) -- (p10)
            -- (p11) -- (p12) -- (p13) -- (p14) -- (p15) -- (p16);
        \end{scope}
      \end{scope}
      \begin{scope}[xshift=1*\xshift]
        \coordss
        \draw[chain] (c0) -- (c1) -- (c2) -- (c3) -- (c4)
          -- (c5) -- (c6) -- (c7) -- (c8) -- (c9) -- (c10)
          -- (c11) -- (c12) -- (c13) -- (c14) -- (c15) -- (c16);
      \end{scope}
    \end{scope}
  \end{tikzpicture}
  \caption{
    The Koch chains $K_s$ for $s=0,\dots,4$ and the corresponding Koch curves.
    Even though it is hard to recognize for larger values of $s$, the changes in direction along the Koch curve on the right are reflected one-to-one by the chain edges of the corresponding Koch chain on the left.
  }
  \label{fig:kochchain}
\end{figure}
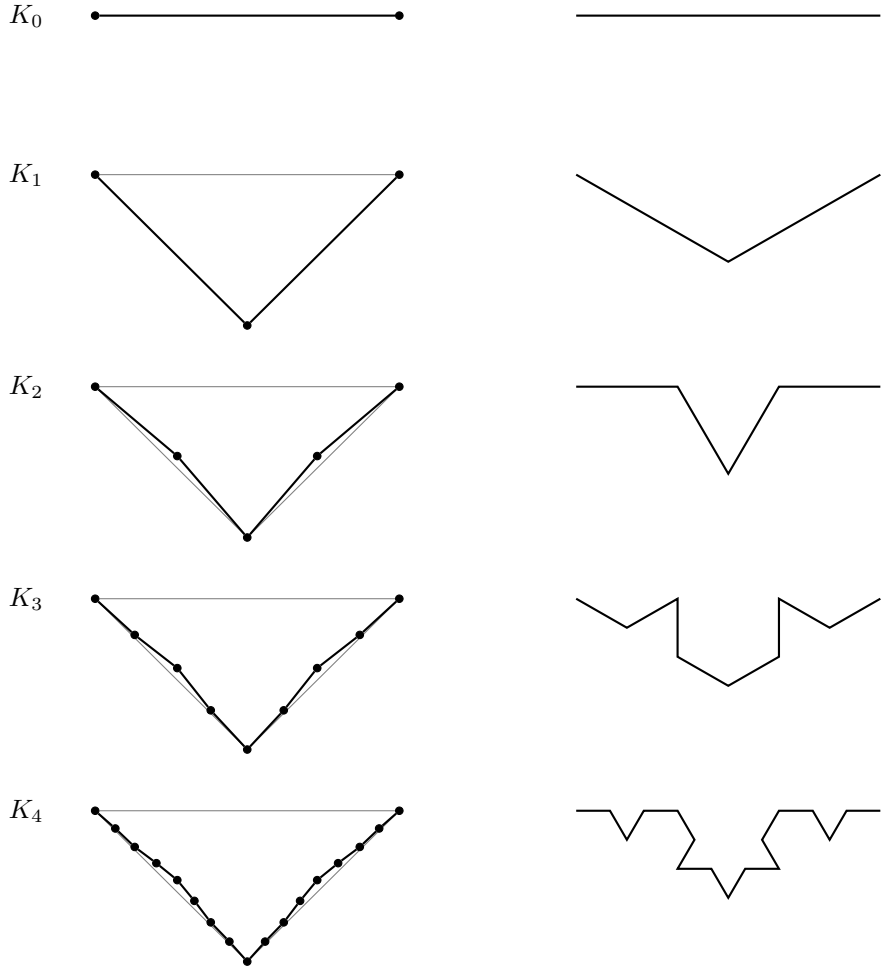

Koch chains turn out to have a particularly large number of triangulations, much more so than any other known point sets.
For values of $s$ up to $21$, we have computed the corresponding numbers of upper and lower triangulations, as well as complete triangulations, by using our algorithm from \cref{thm:algorithm}.
The results are displayed in \cref{tab:kochchainsnumbers}.

\begin{table}
  \centering
  \begin{tabular}{@{}rrlllrrlll@{}}
    \toprule
      $s$ & $n$ & $\sqrt[n]{U}$ & $\sqrt[n]{L}$ & $\sqrt[n]{T}$ & $s$ & $n$ & $\sqrt[n]{U}$ & $\sqrt[n]{L}$ & $\sqrt[n]{T}$ \\
    \cmidrule(r){1-5}\cmidrule(l){6-10}
       0 &       1 & 1.0      & 1.0      & 1.0      & 11 &    2048 & 3.121029 & 2.858643 & 8.921910 \\
       1 &       2 & 1.0      & 1.0      & 1.0      & 12 &    4096 & 2.882177 & 3.121029 & 8.995359 \\
       2 &       4 & 1.189207 & 1.0      & 1.189207 & 13 &    8192 & 3.134955 & 2.882177 & 9.035496 \\
       3 &       8 & 1.791279 & 1.189207 & 2.130201 & 14 &   16384 & 2.889213 & 3.134955 & 9.057554 \\
       4 &      16 & 2.035453 & 1.791279 & 3.646065 & 15 &   32768 & 3.139056 & 2.889213 & 9.069406 \\
       5 &      32 & 2.558954 & 2.035453 & 5.208633 & 16 &   65536 & 2.891256 & 3.139056 & 9.075820 \\
       6 &      64 & 2.564646 & 2.558954 & 6.562814 & 17 &  131072 & 3.140236 & 2.891256 & 9.079229 \\
       7 &     128 & 2.935733 & 2.564646 & 7.529118 & 18 &  262144 & 2.891838 & 3.140236 & 9.081055 \\
       8 &     256 & 2.783587 & 2.935733 & 8.171870 & 19 &  524288 & 3.140569 & 2.891838 & 9.082019 \\
       9 &     512 & 3.075469 & 2.783587 & 8.560839 & 20 & 1048576 & 2.892001 & 3.140569 & 9.082530 \\
      10 &    1024 & 2.858643 & 3.075469 & 8.791671 & 21 & 2097152 & 3.140662 & 2.892001 & 9.082799 \\
    \bottomrule
  \end{tabular}
  \caption{
    The computed numbers of triangulations of the Koch chain $K_s$ for $s=0,\dots,21$, where each entry is rounded down to six decimal places.
    As usual, $n$ is the number of chain edges, whereas $U$, $L$, and $T$ stand, respectively, for the numbers of upper, lower, and complete triangulations of the corresponding Koch chain.
  }
  \label{tab:kochchainsnumbers}
\end{table}

In consequence, concatenating copies of $K_{21}$ side by side results in an infinite family of point sets with at least $9.082799^n$ triangulations.
This alone already establishes the improved lower bound of $\trmax(n) = \Omega(9.082799^n)$.

\subparagraph{Poly chains and Twin chains.}

We were unable to nail down the exact asymptotic behavior of the number of triangulations of $K_s$ as $s$ approaches infinity.
It is also unclear how much is lost due to undercounting by not considering any interactions between the different copies of $K_{21}$ in our simple lower bound construction from just before.

To remedy the situation somewhat, in \cref{sec:polytwin} we define and analyze more carefully the~\emph{poly\nobreakdash-$C_0$ chain} (a specific way of concatenating $N$ copies of a fixed chain $C_0$) and the~\emph{twin\nobreakdash-$C_0$ chain} (a construction where two copies of a poly-$C_0$ chain face each other, similar in spirit to the classic double chain).
Based on these considerations, we get a slightly improved lower bound construction, and we are also able to conclude that the numbers in the last column of \cref{tab:kochchainsnumbers} will not grow significantly larger than what we already have.

\begin{theorem}
  \label{thm:lowerbound}
  Let $C_N$ be the twin-$K_{21}$ chain that is made up of $2N$ copies of $K_{21}$ for a total of $n=2N\cdot 2^{21}+1$ chain edges.
  Then,
  \begin{align*}
    \lim_{N\rightarrow\infty} \sqrt[n]{\tr(C_N)} = 9.083095\dots,
  && \text{and hence }\; \trmax(n) = \Omega(9.083095^n).
  \end{align*}
\end{theorem}

\begin{theorem}
  \label{thm:kochlimit}
  For the Koch chain $K_s$ with $n = 2^s$ chain edges, we have
  \begin{equation*}
    9.082799 \leq \lim_{s\rightarrow\infty}\sqrt[n]{\tr(K_s)} \leq 9.083139.
  \end{equation*}
\end{theorem}

\section{Structural Properties of Chains}
\label{sec:structure}

Recall \cref{def:chain} from the introduction.
Note that the unavoidable chain edges form an $x$-monotone curve $p_0 p_1 \dots p_n$,
to which we refer as the \emph{chain curve}.
An edge $p_i p_j$ that is not a chain edge cannot cross the chain curve,
and so it lies either above or below that curve.

\begin{definition}
To every chain $C$ we associate a \emph{visibility triangle} $V(C)$ with entries
\[
V(C)_{i, j} = \begin{cases}%
  +1, & \text{if $p_i p_j$ lies above the chain curve;}\\%
  -1, & \text{if $p_i p_j$ lies below the chain curve;}\\%
  \phantom{+}0, & \text{if $p_i p_j$ is a chain edge (i.e., $i+1 = j$);}%
\end{cases} \qquad (0 \leq i < j \leq n).
\]
\end{definition}

As an example, the visibility triangles of the chains that correspond to the classic point sets from the introduction can be seen in \cref{fig:visibilitytriangles}.
In all such illustrations, the indices of $V(C)_{i,j}$ are understood in familiar matrix notation;
that is, $i$ indexes the row and $j$ indexes the column.
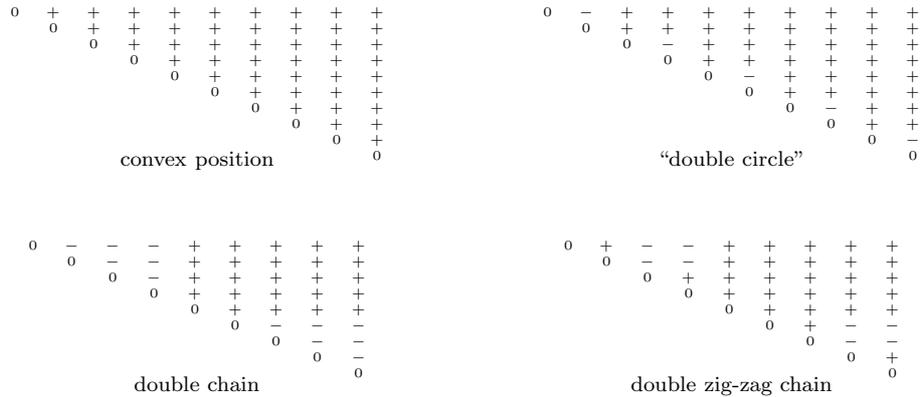
\begin{figure}[b]
  \newcommand{\totalY}{2.5/2}
  \centering
  \begin{tikzpicture}
    \def\xshift{200}
    \def\yshift{85}
    \def\labelshift{1.0}
    \begin{scope}[yshift=-0*\yshift]
      \node {
        \tiny
        $\begin{matrix}
          0 & + & + & + & + & + & + & + & + & + \\
            & 0 & + & + & + & + & + & + & + & + \\
            &   & 0 & + & + & + & + & + & + & + \\
            &   &   & 0 & + & + & + & + & + & + \\
            &   &   &   & 0 & + & + & + & + & + \\
            &   &   &   &   & 0 & + & + & + & + \\
            &   &   &   &   &   & 0 & + & + & + \\
            &   &   &   &   &   &   & 0 & + & + \\
            &   &   &   &   &   &   &   & 0 & + \\
            &   &   &   &   &   &   &   &   & 0 \\
        \end{matrix}$
      };
      \node at (0,-\labelshift) {\footnotesize convex position};
    \end{scope}%
    \begin{scope}[xshift=1*\xshift]
      \node {
        \tiny
        $\begin{matrix}
          0 & - & + & + & + & + & + & + & + & + \\
            & 0 & + & + & + & + & + & + & + & + \\
            &   & 0 & - & + & + & + & + & + & + \\
            &   &   & 0 & + & + & + & + & + & + \\
            &   &   &   & 0 & - & + & + & + & + \\
            &   &   &   &   & 0 & + & + & + & + \\
            &   &   &   &   &   & 0 & - & + & + \\
            &   &   &   &   &   &   & 0 & + & + \\
            &   &   &   &   &   &   &   & 0 & - \\
            &   &   &   &   &   &   &   &   & 0 \\
        \end{matrix}$
      };
      \node at (0,-\labelshift) {\footnotesize \phantom{p}``double circle''\phantom{p}};
    \end{scope}%
    \begin{scope}[yshift=-1*\yshift]
      \node {
        \tiny
        $\begin{matrix}
          0 & - & - & - & + & + & + & + & + \\
            & 0 & - & - & + & + & + & + & + \\
            &   & 0 & - & + & + & + & + & + \\
            &   &   & 0 & + & + & + & + & + \\
            &   &   &   & 0 & + & + & + & + \\
            &   &   &   &   & 0 & - & - & - \\
            &   &   &   &   &   & 0 & - & - \\
            &   &   &   &   &   &   & 0 & - \\
            &   &   &   &   &   &   &   & 0 \\
        \end{matrix}$
      };
      \node at (0,-\labelshift) {\footnotesize \phantom{p}double chain\phantom{p}};
    \end{scope}%
    \begin{scope}[xshift=1*\xshift,yshift=-1*\yshift]
      \node {
        \tiny
        $\begin{matrix}
          0 & + & - & - & + & + & + & + & + \\
            & 0 & - & - & + & + & + & + & + \\
            &   & 0 & + & + & + & + & + & + \\
            &   &   & 0 & + & + & + & + & + \\
            &   &   &   & 0 & + & + & + & + \\
            &   &   &   &   & 0 & + & - & - \\
            &   &   &   &   &   & 0 & - & - \\
            &   &   &   &   &   &   & 0 & + \\
            &   &   &   &   &   &   &   & 0 \\
        \end{matrix}$
      };
      \node at (0,-\labelshift) {\footnotesize double zig-zag chain};
    \end{scope}%
  \end{tikzpicture}
  \caption{
    The visibility triangles corresponding to the chains depicted earlier in \cref{fig:examples}.
    For improved clarity, we only display the signs of the respective entries.
  }
  \label{fig:visibilitytriangles}
\end{figure}

For $i < j < k$, the triangle $p_i p_j p_k$ is oriented counter-clockwise if and only if $p_j$  lies below the edge $p_i p_k$ or, equivalently, if and only if $V(C)_{i, k} = +1$.
It follows that two chains with the same visibility triangle have the same order type and, therefore, the same set of crossing-free geometric graphs and triangulations.
This motivates using the visibility triangle as the canonical representation of a chain.
That is, we consider two chains to be equal if their visibility triangles are identical;
in particular, given two chains $C_1$ and $C_2$, we write $C_1 = C_2$ if $V(C_1) = V(C_2)$.

The edge $p_0 p_n$ plays a crucial role in determining the
shape of a chain. For example, if~$V(C)_{0, n} = +1$, then this edge
is the only edge on the upper convex hull and, from a global
perspective, the chain curve looks like it is curving upwards. Conversely,
if $V(C)_{0, n} = -1$, then the chain curve looks like it is curving downwards.
Correspondingly, we call a chain $C$ with $V(C)_{0, n} \geq 0$ an \emph{upward chain}, and a chain with $V(C)_{0, n} \leq 0$ a \emph{downward chain}.
This implies that every chain is either an upward or a downward chain,
and the primitive chain with only $n=1$ chain edge is the only chain that is both.

\subsection{Flips}

Chains may be flipped upside-down by reflection
at the $x$-axis, thus turning an upward chain into a downward chain, and vice versa.
See \cref{fig:flipexample} for an example.

\begin{figure}
  \centering
  \begin{tikzpicture}
    \def\yshift{80}
    \begin{scope}
      \node at (-1/2,1/2) {$C$}; 
      \node[point] (p0) at (0,0) {};
      \node[point] (p1) at (2/3,1) {};
      \node[point] (p2) at (4/3,1/3) {};
      \node[point] (p3) at (2,0) {};
      \begin{scope}[on background  layer]
        \draw[unavoidable] (p0) -- (p3) -- (p1);
        \draw[chain] (p0) -- (p1) -- (p2) -- (p3);
      \end{scope}
      \node at (3.5,1/2) {
        \footnotesize
        $
          \begin{matrix}
            0 & - & - \\
              & 0 & + \\
              &   & 0 \\
          \end{matrix}
        $
      };
    \end{scope}
    
    \begin{scope}[xshift=192]
      \node at (-1/2,1/2) {$\flip{C}$}; 
      \node[point] (p0) at (0,1) {};
      \node[point] (p1) at (2/3,0) {};
      \node[point] (p2) at (4/3,2/3) {};
      \node[point] (p3) at (2,1) {};
      \begin{scope}[on background  layer]
        \draw[unavoidable] (p0) -- (p3) -- (p1);
        \draw[chain] (p0) -- (p1) -- (p2) -- (p3);
      \end{scope}
      \node at (3.5,1/2) {
        \footnotesize
        $
          \begin{matrix}
            0 & + & + \\
              & 0 & - \\
              &   & 0 \\
          \end{matrix}
        $
      };
    \end{scope}
  \end{tikzpicture}
  \caption{
    A chain $C$ and its flipped version $\flip{C}$ with the corresponding visibility triangles.
  }
  \label{fig:flipexample}
\end{figure}
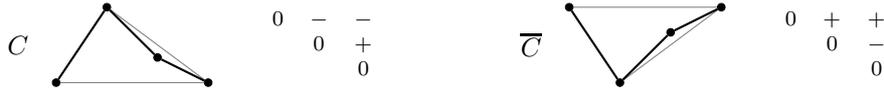

\begin{proposition}
\label{prop:flip}
Let $C$ be a chain with $n$ chain edges.
Then, there is another chain (which we denote by $\flip{C}$ and call the \emph{flipped version of $C$}) with $n$ chain edges and visibility triangle
\[
V(\flip{C})_{i, j} = -V(C)_{i, j} \qquad (0 \leq i < j \leq n).
\]
\end{proposition}

\subsection{Convex and Concave Sums}

Given two chains $C_1$ and $C_2$, we would like to concatenate them so that we get a new chain containing $C_1$ and $C_2$ as substructures.
As shown in \cref{fig:sumexamples}, there are two ways to do so.

\begin{proposition}
\label{prop:convexsum}
Let $C_1$ and $C_2$ be chains with $n_1$ and $n_2$ chain edges, respectively.
Then, there is an upward chain (which we denote by $C_1 \vee C_2$ and call the
\emph{convex sum of $C_1$ and $C_2$}) with $n_1 + n_2$ chain edges
and visibility triangle
\[
V(C_1 \vee C_2)_{i, j} = \begin{cases}%
  V(C_1)_{i, j}, & \text{if $i, j \in \bra{0, n_1}$;}\\%
  V(C_2)_{i-n_1, j-n_1}, & \text{if $i, j \in \bra{n_1, n_1+n_2}$;}\\%
  +1, & \text{if $i < n_1 < j$;}%
\end{cases} \qquad (0 \leq i < j \leq n_1 + n_2).
\]
\end{proposition}
\begin{proposition}
\label{prop:concavesum}
Similarly, there is a downward chain (which we denote by $C_1 \wedge
C_2$ and call the \emph{concave sum of $C_1$ and $C_2$}) with $n_1 + n_2$
chain edges and visibility triangle
\[
V(C_1 \wedge C_2)_{i, j} = \begin{cases}%
  V(C_1)_{i, j}, & \text{if $i, j \in \bra{0, n_1}$;}\\%
  V(C_2)_{i-n_1, j-n_1}, & \text{if $i, j \in \bra{n_1, n_1+n_2}$;}\\%
  -1, & \text{if $i < n_1 < j$;}%
\end{cases} \qquad (0 \leq i < j \leq n_1 + n_2).
\]
\end{proposition}

\begin{figure}
  \centering
  \begin{tikzpicture}
    \def\yshift{75}
    \begin{scope}[xshift=19]
      \node at (-1/2,1/2) {$C_1$}; 
      \node[point] (p0) at (0,0) {};
      \node[point] (p1) at (2/3,1) {};
      \node[point] (p2) at (4/3,1/3) {};
      \node[point] (p3) at (2,0) {};
      \begin{scope}[on background  layer]
        \draw[unavoidable] (p0) -- (p3) -- (p1);
        \draw[chain,myred] (p0) -- (p1) -- (p2) -- (p3);
      \end{scope}
      \node at (3.5,1/2) {
        \footnotesize
        \textcolor{myred}{$
          \begin{matrix}
            0 & - & - \\
              & 0 & + \\
              &   & 0 \\
          \end{matrix}
        $}
      };
    \end{scope}
    
    \begin{scope}[xshift=187]
      \node at (-1/2,1/2) {$C_2$}; 
      \node[point] (p0) at (0,1) {};
      \node[point] (p1) at (2/4,2/3) {};
      \node[point] (p2) at (4/4,0) {};
      \node[point] (p3) at (6/4,2/3) {};
      \node[point] (p4) at (2,1) {};
      \begin{scope}[on background  layer]
        \draw[unavoidable] (p0) -- (p2) -- (p4) -- (p0);
        \draw[chain,myblue] (p0) -- (p1) -- (p2) -- (p3) -- (p4);
      \end{scope}
      \node at (4,1/2) {
        \footnotesize
        \textcolor{myblue}{$
          \begin{matrix}
            0 & - & + & + \\
              & 0 & + & + \\
              &   & 0 & - \\
              &   &   & 0 \\
          \end{matrix}
        $}
      };
    \end{scope}
    
    \begin{scope}[xshift=45,yshift=-1*\yshift]
      \node at (-1,1/2) {$C_1 \vee C_2$}; 
      \node[point] (p0) at (0/7*5,1) {};
      \node[point] (p1) at (1/7*5,2/3+0.1) {};
      \node[point] (p2) at (2/7*5,1/3+0.05) {};
      \node[point] (p3) at (3/7*5,0) {};
      \node[point] (p4) at (4/7*5,0+0.1) {};
      \node[point] (p5) at (5/7*5,0) {};
      \node[point] (p6) at (6/7*5,1/2+0.1) {};
      \node[point] (p7) at (7/7*5,1) {};
      \begin{scope}[on background  layer]
        \draw[unavoidable] (p0) -- (p3) -- (p5) -- (p7);
        \draw[unavoidable] (p0) -- (p7);
        \draw[chain,myred] (p0) -- (p1) -- (p2) -- (p3);
        \draw[chain,myblue] (p3) -- (p4) -- (p5) -- (p6) -- (p7);
      \end{scope}
      \node at (8,1/2) {
        \footnotesize
        \def\red#1{\textcolor{myred}{#1}}
        \def\blu#1{\textcolor{myblue}{#1}}
        $\begin{matrix}
          \red0 & \red- & \red- & + & + & + & + \\
            & \red0 & \red+ & + & + & + & + \\
            &   & \red0 & + & + & + & + \\
            &   &   & \blu0 & \blu- & \blu+ & \blu+ \\
            &   &   &   & \blu0 & \blu+ & \blu+ \\
            &   &   &   &   & \blu0 & \blu- \\
            &   &   &   &   &   & \blu0 \\
        \end{matrix}$
      };
    \end{scope}
    
    \begin{scope}[xshift=45,yshift=-2*\yshift]
      \node at (-1,1/2) {$C_1 \wedge C_2$}; 
      \node[point] (p0) at (0/7*5,0) {};
      \node[point] (p1) at (1/7*5,1/2) {};
      \node[point] (p2) at (2/7*5,3/4-0.1) {};
      \node[point] (p3) at (3/7*5,1) {};
      \node[point] (p4) at (4/7*5,3/4-0.05) {};
      \node[point] (p5) at (5/7*5,1/2-0.15) {};
      \node[point] (p6) at (6/7*5,1/4-0.05) {};
      \node[point] (p7) at (7/7*5,0) {};
      \begin{scope}[on background  layer]
        \draw[unavoidable] (p1) -- (p3) -- (p7);
        \draw[unavoidable] (p0) -- (p7);
        \draw[chain,myred] (p0) -- (p1) -- (p2) -- (p3);
        \draw[chain,myblue] (p3) -- (p4) -- (p5) -- (p6) -- (p7);
      \end{scope}
      \node at (8,1/2) {
        \footnotesize
        \def\red#1{\textcolor{myred}{#1}}
        \def\blu#1{\textcolor{myblue}{#1}}
        $\begin{matrix}
          \red0 & \red- & \red- & - & - & - & - \\
            & \red0 & \red+ & - & - & - & - \\
            &   & \red0 & - & - & - & - \\
            &   &   & \blu0 & \blu- & \blu+ & \blu+ \\
            &   &   &   & \blu0 & \blu+ & \blu+ \\
            &   &   &   &   & \blu0 & \blu- \\
            &   &   &   &   &   & \blu0 \\
        \end{matrix}$
      };
    \end{scope}
  \end{tikzpicture}
  \caption{
    In the top row, two chains $C_1$ and $C_2$ with their visibility
    triangles. Below, the corresponding convex and concave sums
    $C_1 \vee C_2$ and $C_1 \wedge C_2$. Red and blue color is used
    to highlight the contained substructures and their origin.
  }
  \label{fig:sumexamples}
\end{figure}
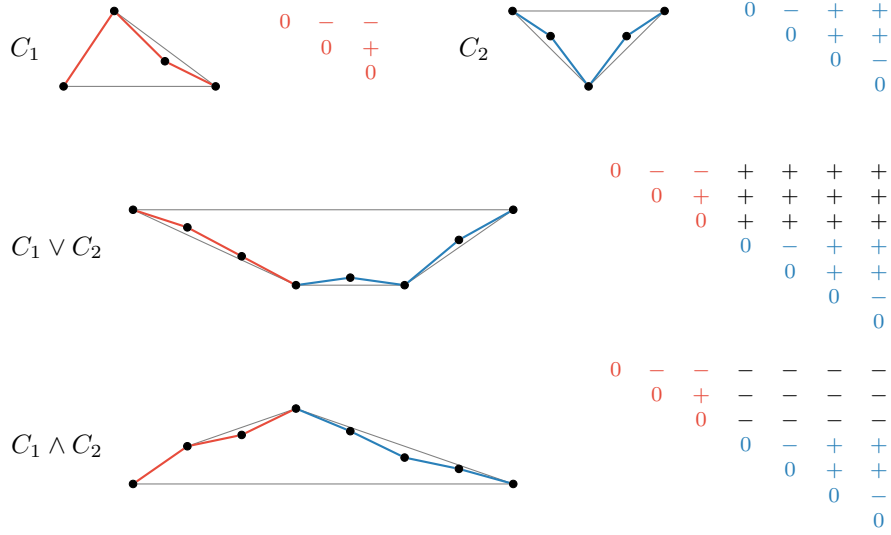

\begin{proof}[Proof of \cref{prop:convexsum}]
We focus on the convex sum;
the proof for the concave sum is analogous.
We have to show that there is a point set
that forms a chain with the specified visibility triangle.
Intuitively speaking, this is achieved by first flattening
the two given chains and then arranging them in a $\vee$-shape.

To be more precise, we employ vertical shearings, which are
maps $(x, y) \mapsto (x, y+\lambda x)$ in $\bR^2$
for some $\lambda \in \bR$.
Vertical shearings preserve
signed areas and $x$-coordinates. Hence, if a point set realizes
a specific chain, then so does its image under any vertical shearing.

With the help of an appropriate vertical shearing, we may realize $C_1$
as a point set in the rectangle $[-1, 0] \times [-1, 1]$ in such a way that
the first point is at $(-1, 0)$ and the last point is at $(0, 0)$.
Then, given any $\varepsilon \geq 0$,
we may rescale vertically to get a point set $Q_1(\varepsilon)$ in
the rectangle $[-1, 0] \times [-\varepsilon, \varepsilon]$.
Let now $R_1(\varepsilon)$ be the image of $Q_1(\varepsilon)$
under the vertical shearing with $\lambda = -1$.
Then, the first point of $R_1(\varepsilon)$ lies at $(-1, 1)$,
while the last point remains at~$(0, 0)$.
For $\varepsilon > 0$, since $Q_1(\varepsilon)$ is a realization of $C_1$,
so is $R_1(\varepsilon)$.
On the other hand, for~$\varepsilon = 0$, the points of $R_1(\varepsilon)$
all lie on the segment between $(-1, 1)$ and $(0, 0)$.

With $C_2$ we proceed similarly to get a point set $Q_2(\varepsilon)$ in
the rectangle $[0, 1] \times [-\varepsilon, \varepsilon]$,
but we now apply the vertical shearing with $\lambda = 1$
to get $R_2(\varepsilon)$ with the first point at $(0, 0)$
and the last point at $(1, 1)$.

Let $T(\varepsilon) = R_1(\varepsilon) \cup R_2(\varepsilon)$.
We claim that for $\varepsilon > 0$ small enough,
$T(\varepsilon)$ is a chain with
visibility triangle $V(C_1 \vee C_2)$ as specified.
Indeed, as $R_i(\varepsilon)$ is a realization of $C_i$,
we only need to check that the edges between any
point of $R_1(\varepsilon)$ and any point of $R_2(\varepsilon)$
(excluding the common point at the origin)
lie above all the points in between. Since this is the case for~$\varepsilon = 0$ and $T(\varepsilon)$ depends
continuously on $\varepsilon$, the claim follows.
\end{proof}

\subsection{Algebraic Properties}

Using the formulas for the visibility triangles from the
corresponding transformations in \cref{prop:flip,prop:convexsum,prop:concavesum}, it
can be checked easily that the following
algebraic laws hold.
\begin{lemma}\label{lem:algebraiclaws}
Let $C_1, C_2, C_3$ be arbitrary chains. Then, the following are all true.
\begin{align*}
  & \text{Involution: } &
    \flip{\flip{C_1}} &= C_1 \\
  & \text{De Morgan: } &
    \flip{C_1 \vee C_2} &= \flip{C_1} \wedge \flip{C_2}, &
    \flip{C_1 \wedge C_2} &= \flip{C_1} \vee \flip{C_2} \\
  & \text{Associativity: } &
    (C_1 \vee C_2) \vee C_3 &= C_1 \vee (C_2 \vee C_3), &
    (C_1 \wedge C_2) \wedge C_3 &= C_1 \wedge (C_2 \wedge C_3)
\end{align*}
\end{lemma}\noindent
However, note that for example $(C_1 \wedge C_2) \vee C_3$ is not the same chain as $C_1 \wedge (C_2 \vee C_3)$.

\subsection{Examples}

We denote by $\cOne$ the primitive chain with only $n=1$ chain edge; that is,
the visibility triangle has just the entry $V(\cOne)_{0,1} = 0$.
Using this as a building block, we may define two more fundamental chains,
the \emph{convex chain $\vex{n}$} and the \emph{concave chain $\cave{n}$}, by setting
\begin{align*}
 \vex{n} &= \underbrace{\cOne \vee \dots \vee \cOne}_{\text{$n$ copies}}, &
 \cave{n} &= \underbrace{\cOne \wedge \dots \wedge \cOne}_{\text{$n$ copies}}.
\end{align*}

The convex chain is an upward chain, while the concave chain
is a downward chain. Also, since $\flip{\cOne} = \cOne$, we
get $\flip{\vex{n}} = \cave{n}$ by using De Morgan's law.
Finally, note that $\vex{n}$ and $\cave{n}$ are distinct as chains, even though
they both describe a set of $n+1$ points in convex position.

As already mentioned in the introduction, many previously studied
point sets are in fact chains, or can be seen as such.
Using flips as well as convex and concave sums, we can now describe
these configurations with very concise formulas.

\begin{example}
For any integer $k \geq 1$, the \emph{double chain} with $n = 2 k + 1$ chain edges is the chain
\[
\cDC{n} = \cave{k} \vee \cOne \vee \cave{k}. 
\]
\end{example}

\begin{example}
For any integer $k \geq 1$, the \emph{zig-zag chain} with $n = 2k$ chain edges (which, in essence, is a double circle with one of the inner points removed) is the chain
\[
\cZZ{n} = \underbrace{\cave{2} \vee \dots \vee \cave{2}}_{\text{$k$ copies}}. 
\]
\end{example}

\begin{example}
For any integer $k \geq 1$, the \emph{double zig-zag chain} with $n = 4k+1$ chain edges is the chain
\[
\cDZZ{n} = \flip{\cZZ{2k}} \vee \cOne \vee \flip{\cZZ{2k}}. 
\]
\end{example}

All these examples involve formulas of constant nesting depth only.
But the tools developed up to this point allow us to also define more
complicated chains via formulas of non-constant nesting depth,
without having to worry about questions of realizability. One such chain
with logarithmic nesting depth is indeed the Koch chain.

\begin{definition}\label{def:kochchain}
The \emph{Koch chain} $\koch{s}$ is an upward chain with $n = 2^s$ chain edges, 
defined recursively via $\koch{0} = \cOne$ and $\koch{s} = \flip{\koch{s-1}} \vee \flip{\koch{s-1}}$ for all integers $s \geq 1$.
\end{definition}

Indeed, after expanding the recursive definition twice and using De Morgan's law
on both sides, we see that the formula
$K_s = (K_{s-2} \wedge K_{s-2}) \vee (K_{s-2} \wedge K_{s-2})$
has a complete binary parse tree with alternating convex and
concave sums on any path from the root to a leaf.

\subsection{Unique Construction}

We want to prove the following result. In essence, it states that
every chain can be constructed in a unique way by using only convex and concave sums.

\begin{theorem} \label{theo:chainformulaexists}
Every chain can be expressed as a formula involving
convex sums, concave sums, parentheses, and copies of the primitive
chain with  only one chain edge. This formula is unique up to
redundant parentheses (redundant due to associativity as in \cref{lem:algebraiclaws}).
\end{theorem}

In particular, the above theorem allows us to encode a chain
with $O(n)$ bits (as opposed to the $O(n^2)$ bits required for the
visibility triangle) and to easily enumerate all chains of a fixed size.
We further see by means of a straightforward bijection to Schr\"oder trees
(these are ordered trees that have no vertex of outdegree one; as such, they
describe precisely the parse trees of the unique formulas of upward chains)
that the number of upward chains
is given by the little Schr\"oder numbers \cite{oeislittleschroeder}
and the number of all chains is
given by the large Schr\"oder numbers \cite{oeisbigschroeder}.

The theorem follows by induction from the
following proposition (and from an analogous proposition
that expresses downward chains as a unique concave sum of upward chains).

\begin{proposition} \label{prop:sumformulaexists}
Let $C$ be an upward chain with $n>1$ chain edges.
Suppose that the lower convex hull of
$C$ is $p_{i_0}p_{i_1}\dots p_{i_k}$ with $0 = i_0 < \dots < i_k = n$.
For $j=1, \dots, k$, let $C_j$ be the chain with points $p_{i_{j-1}}, \dots, p_{i_j}$.
Then, each $C_j$ is a downward chain.
Moreover, $C = C_1 \vee \dots \vee C_k$ and any formula that evaluates to $C$ has the same top-level structure.
\end{proposition}
\begin{proof}
As $p_{i_{j-1}} p_{i_j}$ is an edge of the lower convex hull of $C$,
it is below all the points in between.
Hence, each $C_j$ is indeed a downward chain.

To prove $C = C_1 \vee \dots \vee C_k$,
we have to show that both chains have the same visibility triangle.
By definition of the $C_j$, the visibility triangles
clearly agree on all entries that stem from
an edge $p_a p_b$ where $p_a$ and $p_b$ are both part of the same $C_j$.
On the other hand, if $p_a$ and $p_b$ are not part of the same $C_j$,
then there is a $j$ with $a < i_j < b$. As $p_{i_j}$ is a vertex
of the lower convex hull, it lies below the edge $p_a p_b$ and
hence $V(C)_{a, b} = +1$. But this is precisely what we also get
for the visibility triangle of the convex sum $C_1 \vee \dots \vee C_k$.

For uniqueness, suppose we are given any formula for $C$.
Since $C$ is assumed to be an upward chain and since any concave sum is a downward chain,
the formula must be of the form $C'_1 \vee \dots \vee C'_{k'}$. We may
further assume that each $C'_j$ is a downward chain by omitting redundant
parentheses. Observe now that in any such convex sum of downward chains,
the resulting lower convex hull is determined by the points that are shared by any two
consecutive chains $C'_j$. Since the given formula evaluates to $C$, we
must have $k' = k$ and $C'_j = C_j$ for all $j$.
\end{proof}

\subsection{Geometric Characterization}

As already mentioned in the introduction, the chain edges together with the hull edge $p_0 p_n$ form a spanning cycle of unavoidable edges.
We are now ready to prove that this property characterizes chains geometrically in terms of order types.

Suppose we are given a set of $n+1$ points $p_0,p_1,\dots,p_n$ as in \cref{thm:characterization}.
Without loss of generality, let $\SC = p_0 p_1 \dots p_n$ be the spanning cycle of unavoidable edges in counter-clockwise order, with $p_0p_n$ an edge of the convex hull, which we call the \emph{base edge}.
As $\SC$ consists of unavoidable edges only, it cannot be crossed by any edge that is not part of $\SC$.
Hence, we can associate a visibility triangle with the given point set, similar to the visibility triangle of a chain, by setting
\[
V_{i, j} = \begin{cases}%
  +1, & \text{if $p_i p_j$ is inside $\SC$ or the base edge;} \\%
  -1, & \text{if $p_i p_j$ is outside $\SC$;} \\%
  \phantom{+}0, & \text{if $p_i p_j$ is part of $\SC$ (i.e., $i+1 = j$);}%
\end{cases} \qquad (0 \leq i < j \leq n).
\]

\begin{proposition}\label{prop:geometricconsiderations}
  For $i < j < k$, the triangle $p_i p_j p_k$ is oriented counter-clockwise if and only if we have $V_{i, k} = +1$.
\end{proposition}

\begin{proof}
We will make use of the following key fact about simple polygons.
Consider two simple polygons that share an edge $e$, but do not touch or intersect in any other way.
Then, exactly one of the following is true.
Either one polygon is contained in the other, or the two polygons touch the edge~$e$ from different sides.

We split the proof of the proposition into the following three cases.

\begin{itemize}
\item Suppose that $p_ip_k$ is the base edge (i.e., $i = 0$ and $k = n$).
Then, the triangle $p_i p_j p_k$ is always oriented counter-clockwise since $p_i p_k = p_0 p_n$ is an edge of the convex hull.
Also, we have $V_{i,k} = V_{0, n} = +1$ by definition.
\item Suppose that $p_i p_k$ is not the base edge (i.e., $i \neq 0$ or $k \neq n$) and inside $\SC$ (i.e., $V_{i,k} = +1$).
Let $\SC'$ denote the simple polygon $p_0, p_1, \dots, p_i, p_k, p_{k+1}, \dots, p_n$.
Let us now apply our key fact about simple polygons to $\SC'$ and the triangle $p_i p_j p_k$.
These two polygons share the edge $p_i p_k$, but do not touch or intersect in any other way since all other edges of $\SC'$ are unavoidable.
Moreover, since $p_j$ is outside of $\SC'$, the triangle is not contained in $\SC'$.
On the other hand, $\SC'$ contains the hull vertices $p_0$ and $p_n$, and so $\SC'$ is also not contained in the triangle.
Therefore, since $\SC'$ touches the edge $p_i p_k$ from the left side, $p_j$ has to be to the right of $p_i p_k$, and so the triangle $p_i p_j p_k$ is indeed oriented counter-clockwise.
\item Suppose that $p_i p_k$ is not the base edge (i.e., $i \neq 0$ or $k \neq n$) and outside $\SC$ (i.e., $V_{i,k} = -1$).
Let again $\SC'$ denote the simple polygon $p_0, p_1, \dots, p_i, p_k, p_{k+1}, \dots, p_n$, and let us again apply our key fact to $\SC'$ and $p_i p_j p_k$.
In this case, since $p_j$ lies inside of $\SC'$, the triangle is contained in $\SC'$.
Therefore, since $\SC'$ touches the edge $p_i p_k$ from the left side, $p_j$ has to be to the left of $p_i p_k$, and so the triangle $p_i p_j p_k$ is indeed oriented clockwise. \qedhere
\end{itemize}
\end{proof}

Having \cref{prop:geometricconsiderations} at hand, it will now suffice to construct a chain whose visibility triangle agrees with $V$ in order to finish the proof of \cref{thm:characterization}.

Let thus $p_{i_0}, p_{i_1}, \dots, p_{i_k}$ be the vertices of the
convex hull with $0 = i_0 < \dots < i_k = n$.
For $1 \leq j \leq k$, let $P_j = \{p_{i_{j-1}}, \dots, p_{i_j}\}$.
We now see that either $P_j$ consists of only two points or that it admits a spanning
cycle of unavoidable edges,
namely $\SC_j = p_{i_{j-1}} p_{i_{j-1}+1} \dots p_{i_{j}}$
with base edge $p_{i_{j-1}} p_{i_{j}}$.
The situation is depicted in \cref{fig:characterization}.
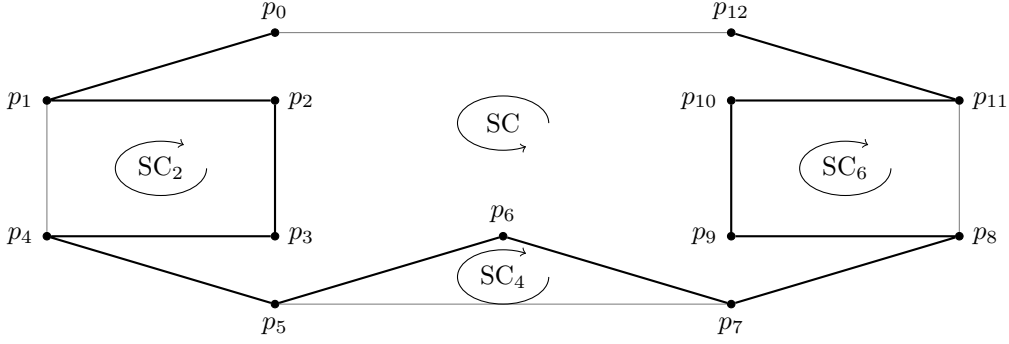
\begin{figure}
  \centering
  \begin{tikzpicture}[xscale=3,yscale=1.8]
    \node[point,label=above:$p_0$] (p0) at (-1,0) {};
    
    \node[point,label=left:$p_1$] (p1) at (-2,-1/2) {};
    \node[point,label=right:$p_2$] (p2) at (-1,-1/2) {};
    \node[point,label=right:$p_3$] (p3) at (-1,-3/2) {};
    \node[point,label=left:$p_4$] (p4) at (-2,-3/2) {};
    
    \node[point,label=below:$p_5$] (p5) at (-1,-2) {};
    \node[point,label=above:$p_6$] (p6) at (0,-3/2) {};
    \node[point,label=below:$p_7$] (p7) at (1,-2) {};
    
    \node[point,label=right:$p_8$] (p8) at (2,-3/2) {};
    \node[point,label=left:$p_9$] (p9) at (1,-3/2) {};
    \node[point,label=left:$p_{10}$] (p10) at (1,-1/2) {};
    \node[point,label=right:$p_{11}$] (p11) at (2,-1/2) {};
    
    \node[point,label=above:$p_{12}$] (p12) at (1,0) {};
    
    \begin{scope}[on background layer]
      \draw[unavoidable] (p1) -- (p4);
      \draw[unavoidable] (p5) -- (p7);
      \draw[unavoidable] (p8) -- (p11);
      \draw[unavoidable] (p12) -- (p0);
      \draw[chain] (p0) -- (p1) -- (p2) -- (p3) -- (p4)
        -- (p5) -- (p6) -- (p7)
        -- (p8) -- (p9) -- (p10) -- (p11) -- (p12);
    \end{scope}
      
    \node at (0,-2/3) {$\SC$};
    \draw[->] (0+1/5,-2/3) arc (0:300:1/5);
    \node at (-3/2,-1) {$\SC_2$};
    \draw[->] (-3/2+1/5,-1) arc (0:-300:1/5);
    \node at (0,-9/5) {$\SC_4$};
    \draw[->] (0+1/5,-9/5) arc (0:-300:1/5);
    \node at (3/2,-1) {$\SC_6$};
    \draw[->] (3/2+1/5,-1) arc (0:-300:1/5);
  \end{tikzpicture}
  \caption{
    The situation in the proof of \cref{thm:characterization}.
    Beware that this is just a sketch;
    in reality, the pockets would need to be much more narrow in order
    to make all edges of $\SC$ unavoidable.
  }
  \label{fig:characterization}
\end{figure}
Note that the inside of $\SC_j$ is
outside of $\SC$. In fact, $\SC_j$ forms a so-called
pocket of $\SC$, which means that all edges of the cycle $\SC_j$ except for $p_{i_{j-1}} p_{i_{j}}$
are also edges of $\SC$.

Using the statement of \cref{thm:characterization} inductively for the strictly smaller set of points $P_j$,
we now see that there is a chain $C_j$
with the same order type as $P_j$, that is, with
\[
V(C_j)_{a, b} = \begin{cases}%
  +1, & \text{if $p_{i_{j-1} + a} p_{i_{j-1} + b}$ is inside $\SC_j$ or the base edge;} \\%
  -1, & \text{if $p_{i_{j-1} + a} p_{i_{j-1} + b}$ is outside $\SC_j$;} \\%
  \phantom{+}0, & \text{if $p_{i_{j-1} + a}p_{i_{j-1} + b}$ is part of $\SC_j$ (i.e., $a+1 = b$);}%
\end{cases} \hspace{-1.5cm} \qquad (0 \leq a < b \leq i_j - i_{j-1}).
\]

Let us consider the flipped version $\flip{C_j}$.
As noted before, the inside of $\SC_j$ is outside of $\SC$.
As $\SC_j$ moreover forms a pocket of $\SC$,
any edge outside of $\SC_j$ is inside $\SC$.
Hence,
\[
V(\flip{C_j})_{a, b} = \begin{cases}%
  +1, & \text{if $p_{i_{j-1}+a} p_{i_{j-1}+b}$ is inside $\SC$;} \\%
  -1, & \text{if $p_{i_{j-1}+a} p_{i_{j-1}+b}$ is outside $\SC$;} \\%
  \phantom{+}0, & \text{if $p_{i_{j-1}+a} p_{i_{j-1}+b}$ is part of $\SC$ (i.e., $a+1 = b$);}
\end{cases} \hspace{-1.5cm} \qquad (0 \leq a < b \leq i_j - i_{j-1}).
\]

We claim that $C = \flip{C_1} \vee \dots \vee \flip{C_k}$
has the desired visibility triangle $V$.
We have just seen that the entries stemming from the
individual $\flip{C_j}$ are correct.
So, all that is left to observe is that edges between
different pockets lie inside of $\SC$, which is indeed the case.

\subsection{Point Sets that are Almost Chains}

Recall the assumption in \cref{thm:characterization} which says that at least one of the edges of the unavoidable spanning cycle must be a convex hull edge.
This assumption is crucial since, for example, the classic double circle cannot be realized as a chain even though it admits an unavoidable spanning cycle.
Nevertheless, we can show that the corresponding order types are reasonably close to being a chain, in the sense that a single point may be replaced by two new points in order to obtain a chain.
In addition, this transformation does not change the number of triangulations by much.

\begin{theorem}
  Let $P$ be a set of $n$ points that admits an unavoidable spanning cycle $p_1,\dots,p_n$ where $p_n$ is a vertex of the convex hull.
  Then, there exists a chain $C= q_0,q_1,\dots,q_n$ such that $p_1,\dots,p_{n-1}$ has the same order type as $q_1,\dots,q_{n-1}$.
  Moreover, we have
  \[
    2 \cdot \tr(P) \leq \tr(C) \leq (n-1) \cdot \tr(P).
  \]
\end{theorem}

\begin{proof}
  We mimic the proof of \cref{thm:characterization} from the previous subsection.
  That is, consider the augmented sequence $p_0,p_1,\dots,p_n$ with $p_0 = p_n$ and let $p_{i_0},\dots,p_{i_k}$ be the vertices of the convex hull with $0 = i_0 < \dots < i_k = n$ and $k \geq 3$.
  For $1 \leq j \leq k$, let $P_j = \{p_{i_{j-1}},\dots,p_{i_j}\}$.
  Each such $P_j$ again forms a so-called pocket, which may be realized as a chain $C_j$.
  The construction is concluded by defining $C = C_1 \vee \dots \vee C_k$.
  Observe by using a statement analogous to \cref{prop:geometricconsiderations} that all involved triangles maintain their orientation, except that the duplicated points $q_0$ and $q_n$ that came from $p_0 = p_n$ behave antagonistically with respect to points in $C_k$ and $C_1$, respectively.
  For example, the triangle $q_iq_jq_n$ for $q_i,q_j \in C_1$ with $i<j$ is always oriented counter-clockwise, even if originally $p_ip_jp_n$ was oriented clockwise.
  
  Let us fix any triangulation of $P$.
  We show how to map it injectively to two distinct triangulations of $C$, thus proving the first inequality.
  First of all, note that the triangulations of the individual pockets of $P$ directly correspond to a lower triangulation of $C$, which means that we may limit our attention to the upper triangulation of $C$.
  There, we start by either adding the triangle $q_0q_1q_n$ or the triangle $q_0q_{n-1}q_n$.
  In the first case, we are left to triangulate the simple polygon $q_1,\dots,q_n$,
  whereas in the second case, the respective polygon is $q_0,\dots,q_{n-1}$.
  In the first case, any remaining edge $p_ip_j$ of the fixed triangulation of $P$ (i.e., any edge that is not contained inside of a pocket) is simply mapped to the edge $q_iq_j$ in $C$.
  In the second case, we do the same except that the point $q_0$ takes the role of $q_n$.
  
  Let us fix any triangulation of $C$.
  To prove the second inequality, we construct a corresponding triangulation of $P$ in such a way that at most $n-1$ triangulations of $C$ have the same correspondence.
  First of all, we may again map the lower triangulation of $C$ to triangulations of the individual pockets of $P$.
  As far as the upper triangulation of $C$ is concerned, there is exactly one triangle of the form $q_0q_mq_n$ where $1 \leq m \leq n-1$.
  This triangle is mapped to the edge $p_np_m$ in the corresponding triangulation of $P$.
  This has the effect of splitting the inside of the unavoidable spanning cycle of $P$ into two separate simple polygons that are left to be triangulated.
  They are triangulated simply by mapping all the remaining edges $q_iq_j$ of the upper triangulation of $C$ to edges $p_ip_j$, where any appearance of $q_0$ is mapped to $p_n$.
\end{proof}

\section{Triangulations of Chains}
\label{sec:triangulations}

In the previous section, we have seen that any chain can be expressed
as a formula involving only convex and concave sums.
Our goal here is to understand how triangulations behave with respect
to such convex and concave sums.
In order for this to work out, we have to consider not just triangulations,
but a more general notion of partial triangulations.

We start by decomposing triangulations of a chain $C$
into an upper and a lower part.
An edge $p_i p_j$ is an \emph{upper edge}
if $V(C)_{i, j} = +1$, a \emph{chain edge} if $V(C)_{i, j} = 0$,
and a \emph{lower edge} if~$V(C)_{i, j} = -1$.
That is, upper edges lie above the chain curve,
while lower edges lie below.

\begin{definition}
An \emph{upper (lower) triangulation} of a given chain $C$ is
a crossing-free geometric graph on $C$ that is edge-maximal
subject to only containing chain edges and upper (lower) edges.
We denote the number of upper and lower triangulations by
$U(C)$ and $L(C)$, respectively, and as always the number of (complete)
triangulations by $\TR(C)$.
\end{definition}

Note that the chain edges are contained in every upper and lower
triangulation. Moreover, every triangulation
is the union of a unique upper and a unique lower triangulation,
which implies $\TR(C) = U(C) \cdot L(C)$.
A lower triangulation of a chain $C$ is an upper triangulation
of the flipped version $\flip{C}$, and therefore $L(C) = U(\flip{C})$.
For this reason, we may restrict our attention to studying only
upper triangulations.

Intuitively speaking, we can create a partial upper
triangulation by combining all the chain edges with some upper
edges, in such a way that all bounded faces are triangles.
Note that then, only some of the used edges are visible from above.

\begin{definition}
Let $C$ be any chain with $n$ chain edges, and let $V = p_{i_0} p_{i_1} \dots p_{i_v}$
with $0 = i_0 < i_1 < \dots < i_v = n$
be an (x-monotone) curve composed
of chain edges and upper edges only.
A \emph{partial upper triangulation of $C$ (with visible edges $V$)}
consists of all chain edges, all edges in $V$,
and a triangulation of the areas between the two.
\end{definition}

\Cref{fig:partialexamples} depicts some partial upper triangulations and their visible edges.
We are interested in counting such triangulations
parameterized by the number of triangles.
It can be noted that a partial upper triangulation with $k$ triangles
has $n-k$ visible edges.

\begin{figure}
  \centering
  \begin{tikzpicture}
    \newcommand{\totalY}{2.5/2}
    \def\xshift{90}
    \newcommand{\rad}{\totalY}
    \newcommand{\radd}{0.93*\rad}
    \newcommand{\ang}{360/10}
    \def\drawdoublecircle{
      \node[point] (p0) at (-0*\ang:\rad) {};
      \node[point] (p1) at (-1*\ang:\rad) {};
      \node[point] (p2) at (-2*\ang:\rad) {};
      \node[point] (p3) at (-3*\ang:\rad) {};
      \node[point] (p4) at (-4*\ang:\rad) {};
      \node[point] (p5) at (-5*\ang:\rad) {};
      \node[point] (q0) at (-0.5*\ang:\radd) {};
      \node[point] (q1) at (-1.5*\ang:\radd) {};
      \node[point] (q2) at (-2.5*\ang:\radd) {};
      \node[point] (q3) at (-3.5*\ang:\radd) {};
      \node[point] (q4) at (-4.5*\ang:\radd) {};
    }
    \begin{scope}[xshift=0*\xshift]
      \drawdoublecircle
      \begin{scope}[on background layer]
        \draw[chain,myblue] (p0) -- (q0) -- (p1) -- (q1) -- (p2) -- (q2)
          -- (p3) -- (q3) -- (p4) -- (q4) -- (p5);
      \end{scope}
    \end{scope}%
    \begin{scope}[xshift=1*\xshift]
      \drawdoublecircle
      \begin{scope}[on background layer]
        \draw[edge,myblue,semithick] (q0) -- (p2);
        \draw[edge,myblue,semithick] (q2) -- (q3) -- (q4);
        \draw[edge] (q0) -- (q1);
        \draw[chain,myblue] (p0) -- (q0);
        \draw[chain,myblue] (p2) -- (q2);
        \draw[chain,myblue] (q4) -- (p5);
        \draw[chain] (q0) -- (p1) -- (q1) -- (p2);
        \draw[chain] (q2) -- (p3) -- (q3) -- (p4) -- (q4);
      \end{scope}
    \end{scope}%
    \begin{scope}[xshift=2*\xshift]
      \drawdoublecircle
      \begin{scope}[on background layer]
        \draw[edge,myblue,semithick] (p0) -- (q1) -- (q2) -- (p5);
        \draw[edge] (q0) -- (q1);
        \draw[edge] (q2) -- (q3) -- (q4) -- (q2);
        \draw[chain] (p0) -- (q0) -- (p1) -- (q1) -- (p2) -- (q2)
          -- (p3) -- (q3) -- (p4) -- (q4) -- (p5);
      \end{scope}
    \end{scope}%
    \begin{scope}[xshift=3*\xshift]
      \drawdoublecircle
      \begin{scope}[on background layer]
        \draw[edge,myblue,semithick] (p5) -- (p0);
        \draw[edge] (p0) -- (q3);
        \draw[edge] (p0) -- (q4);
        \draw[edge] (q4) -- (q3);
        \draw[edge] (q0) -- (q1) -- (q3) -- (q0);
        \draw[edge] (p3) -- (q1) -- (q2);
        \draw[chain] (p0) -- (q0) -- (p1) -- (q1) -- (p2) -- (q2)
          -- (p3) -- (q3) -- (p4) -- (q4) -- (p5);
      \end{scope}
    \end{scope}%
  \end{tikzpicture}
  \caption{
    Four partial upper triangulations of the ``double circle'' with ten, six, three, and one visible edge, respectively.
    As usual, chain edges are in bold, while visible edges are in blue.
  }
  \label{fig:partialexamples}
\end{figure}
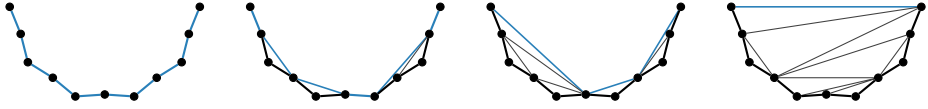

\begin{definition}
Let $C$ be any chain with $n$ chain edges.
For $k = 0,\dots,n-1$, let $t_k(C)$
be the number of partial upper triangulations of $C$ with $k$ triangles (i.e., with $n-k$ visible edges).
The \emph{upper triangulation polynomial of $C$} is the corresponding generating function
\[
T_C(x) = \sum_{k=0}^{n-1} t_k(C) x^k.
\]
\end{definition}

As an example, enumerating all partial upper triangulations
of the convex chain $\vex{4}$ shows that
$T_{\vex{4}}(x) = 1 + 3 x + 5 x^2 + 5 x^3$.
In general, note that for every chain $C$ we have $t_0(C) = 1$ and that
the leading coefficient of $T_{C}(x)$ is equal to $U(C)$.
Moreover, we may again think of $T_{\flip{C}}(x)$ as the
``lower triangulation polynomial'' of $C$.

\subsection{Convex and Concave Sums}

Let us start with the easy case.
For concave sums, we can establish the following relation.

\begin{lemma}\label{lem:concavesum}
A partial upper triangulation of $C_1 \wedge C_2$
is the union of a unique partial upper triangulation of $C_1$
and a unique partial upper triangulation of $C_2$. Hence,
\begin{align*}
T_{C_1 \wedge C_2}(x) = T_{C_1}(x) \cdot T_{C_2}(x),
&&
U(C_1 \wedge C_2) = U(C_1) \cdot U(C_2).
\end{align*}
\end{lemma}

Convex sums are more tricky.
The main insight is that every partial upper triangulation
of~$C_1 \vee C_2$ consists of a partial upper triangulation
of~$C_1$, a partial upper triangulation of~$C_2$,
and some edges between $C_1$ and $C_2$. More precisely:

\begin{proposition}\label{prop:convexsumpartialtriangulation}
There is a triangle-preserving bijection between
\begin{itemize}
\item all triples $(T_1, T_2, T_3)$ where $T_1$ is a partial upper triangulation of $C_1$
(with $v_1$ visible edges), $T_2$ is a partial upper triangulation of $C_2$
(with $v_2$ visible edges), and $T_3$ is a partial upper triangulation of
the convex sum $\cave{v_1} \vee \cave{v_2}$, and
\item all partial upper triangulations of $C_1 \vee C_2$.
\end{itemize}
\end{proposition}

\begin{figure}
  \centering
  \begin{tikzpicture}
    \newcommand{\totalY}{2.5/2}
    \def\xshift{90}
    \newcommand{\rad}{\totalY}
    \newcommand{\radd}{0.89*\rad}
    \newcommand{\ang}{360/8}
    \def\drawleft{
      \node[point] (p0) at (-4*\ang:\rad) {};
      \node[point] (p1) at (-3.5*\ang:\radd) {};
      \node[point] (p2) at (-3*\ang:\rad) {};
      \node[point] (p3) at (-2.5*\ang:\radd) {};
    }
    \def\drawright{
      \node[point] (q0) at (-0*\ang:\rad) {};
      \node[point] (q1) at (-0.5*\ang:\radd) {};
      \node[point] (q2) at (-1*\ang:\rad) {};
      \node[point] (q3) at (-1.5*\ang:\rad) {};
    }
    \def\drawbot{
      \node[point] (m) at (-2*\ang:\rad) {};
    }
    \begin{scope}[xshift=-0.3*\xshift]
      \drawleft
      \drawbot
      \begin{scope}[on background layer]
        \draw[edge,myblue,semithick] (p3) -- (p0);
        \draw[chain,myblue] (p3) -- (m);
        \draw[edge] (p1) -- (p3);
        \draw[chain] (p0) -- (p1) -- (p2) -- (p3);
      \end{scope}
      \node at (-0.5,-1.75) {$C_1$};
    \end{scope}%
    \begin{scope}[xshift=0*\xshift]
      \drawright
      \drawbot
      \begin{scope}[on background layer]
        \draw[chain,myred] (q0) -- (q1) -- (q2);
        \draw[edge,myred,semithick] (m) -- (q2);
        \draw[chain] (q2) -- (q3) -- (m);
      \end{scope}
      \node at (0.5,-1.75) {$C_2$};
    \end{scope}%
    \begin{scope}[xshift=1*\xshift]
      \node[point] (p0) at (-\rad, 0) {};
      \node[point] (p3) at (-0.45*\rad, -0.45*\rad) {};
      \node[point] (m) at (0, -\rad) {};
      \node[point] (q2) at (0.29*\rad, -0.64*\rad) {};
      \node[point] (q1) at (0.64*\rad, -0.29*\rad) {};
      \node[point] (q0) at (\rad, 0) {};
      \begin{scope}[on background layer]
        \draw[chain,myblue] (p0) -- (p3) -- (m);
        \draw[chain,myred] (m) -- (q2) -- (q1) -- (q0);
        \draw[edge] (q2) -- (p3) -- (q1)    (p3) -- (q0);
      \end{scope}
      \node at (0,-1.75) {$\cave{2} \vee \cave{3}$};
    \end{scope}%
    \begin{scope}[xshift=2.2*\xshift]
      \drawleft
      \drawbot
      \drawright
      \begin{scope}[on background layer]]
        \draw[edge,myblue,semithick] (p3) -- (p0);
        \draw[chain,myblue] (p3) -- (m);
        \draw[edge] (p1) -- (p3);
        \draw[chain] (p0) -- (p1) -- (p2) -- (p3);
        \draw[chain,myred] (q0) -- (q1) -- (q2);
        \draw[edge,myred,semithick] (m) -- (q2);
        \draw[chain] (q2) -- (q3) -- (m);
        \draw[edge] (q2) -- (p3) -- (q1)    (p3) -- (q0);
      \end{scope}
      \node at (0,-1.75) {$C_1 \vee C_2$};
    \end{scope}%
  \end{tikzpicture}
  \caption{
    From left to right, the respective partial upper triangulations $T_1$ of $C_1$, $T_2$ of $C_2$, $T_3$ of $\cave{2} \vee \cave{3}$, and the resulting partial upper triangulation of $C_1 \vee C_2$ as in \cref{prop:convexsumpartialtriangulation}.
  }
  \label{fig:concavesumdecomp}
\end{figure}
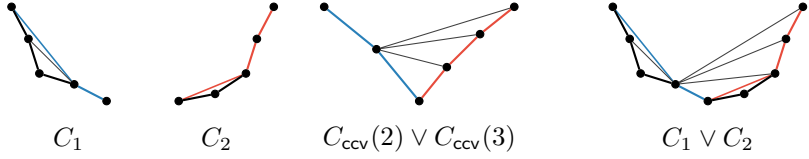

This bijection is defined by taking the union of all triangles,
see \cref{fig:concavesumdecomp}.
The proposition then directly implies the following equation for
the upper triangulation polynomial.
\begin{lemma} \label{lem:convexsum}
Let $C_1$ and $C_2$ be chains with $n_1$ and $n_2$ chain edges, respectively. Then,
\[
T_{C_1 \vee C_2}(x) = \sum_{k_1=0}^{n_1-1} \sum_{k_2=0}^{n_2-1} t_{k_1}(C_1) \cdot t_{k_2}(C_2) \cdot x^{k_1+k_2} \cdot T_{\cave{n_1 - k_1} \, \vee\, \cave{n_2 - k_2}}(x).
\]
\end{lemma}

Let us consider the special case of a convex sum of two concave chains with $n_1$ and $n_2$ chain edges, respectively.
Note that any partial upper triangulation of such a chain
has at most one upper edge that is visible. Summing over all
possibilities for that edge and using the formula from the analysis of the classic double chain from~\cite{garcia2000}, we get
\begin{align*}
T_{\cave{n_1} \, \vee\, \cave{n_2}}(x) = 1 + \sum_{l=1}^{n_1} \sum_{r=1}^{n_2} \binom{l+r-2}{l-1} x^{l+r-1}.
\end{align*}

Combining the above equation with \cref{lem:convexsum} allows us to compute $T_{C_1 \vee C_2}(x)$
from $T_{C_1}(x)$ and $T_{C_2}(x)$.
Furthermore, by comparing the coefficients in the formulas from \cref{lem:concavesum,lem:convexsum}, we get the following obvious but important fact.
\begin{corollary} \label{coro:vexmorecave}
$C_1 \vee C_2$ has at least as many (partial) upper triangulations as $C_1 \wedge C_2$. That is,
\begin{align*}
\forall k\colon t_k(C_1 \vee C_2) \geq t_k(C_1 \wedge C_2),&& U(C_1 \vee C_2) \geq U(C_1 \wedge C_2).
\end{align*}
\end{corollary}

Finally, note that the two chains $C_1 \vee C_2$ and $C_2 \vee C_1$
can be quite different from a geometric point of view. But in terms of the
number of triangulations, they are the same.
\begin{corollary} \label{coro:triangcommutative}
For any two chains $C_1$ and $C_2$, we have
\begin{align*}
T_{C_1 \vee C_2}(x) = T_{C_2 \vee C_1}(x), && T_{C_1 \wedge C_2}(x) = T_{C_2 \wedge C_1}(x).
\end{align*}
\end{corollary}

\subsection{Dynamic Programming}

In this subsection, we show how to use dynamic programming in order to speed up
the computations for a convex sum.
To simplify the analysis, we assume a computational model where all
additions and multiplications take only constant time.

\begin{proposition} \label{prop:convexsumdp}
Let $C_1$ and $C_2$ be chains with $n_1$ and $n_2$ chain edges, respectively.
Given the coefficients of $T_{C_1}(x)$ and $T_{C_2}(x)$,
we can compute $T_{C_1 \vee C_2}(x)$ in $O(n_1 n_2)$ time.
\end{proposition}

Recall that by \cref{theo:chainformulaexists}, we can write any chain
$C$ as a formula involving only convex sums, concave sums, and primitive
chains with only one chain edge.
Therefore, using \cref{prop:convexsumdp} for convex sums and
\cref{lem:concavesum} for concave sums, we are able to compute $T_C(x)$
in quadratic time.
Clearly, this proves \cref{thm:algorithm} from the introduction.

\begin{proof}[Proof of \cref{prop:convexsumdp}]
Observe that every partial upper triangulation of $C_1 \vee C_2$
either corresponds to a partial upper triangulation of $C_1 \wedge C_2$,
or it has a unique visible upper edge that connects a vertex of $C_1$
with a vertex of $C_2$. Let us call this edge the \emph{bridge}.
Let further~$\DP[l][r]$ be the number of partial upper triangulations
whose visible edges consist of~$l$ visible edges in $C_1$, followed by the bridge, followed by~$r$ visible edges in $C_2$.
Then,
\[
T_{C_1 \vee C_2}(x) = T_{C_1 \wedge C_2}(x) + \sum_{l=0}^{n_1-1} \sum_{r=0}^{n_2-1} \DP[l][r] \cdot x^{n_1 + n_2-l-r-1}.
\]

To compute the table $\DP$, let us see what happens when we remove the bridge.
We either end up with a partial upper triangulation of $C_1 \wedge C_2$
with $l+1$ and $r+1$ visible edges in $C_1$ and $C_2$, respectively,
or we get a new bridge, which used to be an edge of the triangle below the old bridge.
In the latter case, depending on which of the two possible edges this is,
we end up with one more visible edge in either $C_1$ or $C_2$.
\Cref{fig:bridgethreecases} depicts these three cases. To summarize,
for all $l$ and $r$ ($0 \leq l < n_1, 0 \leq r < n_2$),
\[
\DP[l][r] = t_{n_1-l-1}(C_1) \cdot t_{n_2-r-1}(C_2) + \DP[l+1][r] + \DP[l][r+1],
\]
with the base case $\DP[n_1][r] = \DP[l][n_2] = 0$. Therefore,
filling up the table $\DP$ takes $O(n_1 n_2)$ time, as desired.
\end{proof}

\begin{figure}
  \centering
  \begin{tikzpicture}[xscale=0.5,yscale=0.15]
    \def\xshift{250}
    \def\drawchain{
      \node[point] (p0) at (-3,7/2) {};
      \node[point] (p1) at (-2,0) {};
      \node[point] (p2) at (-1,-1/4) {};
      \node[point] (p3) at (0,-1) {};
      \node[point] (p4) at (1,-1/4) {};
      \node[point] (p5) at (2,0) {};
      \node[point] (p6) at (3,7/2) {};
      \draw[chain] (p0) -- (p1) -- (p2) -- (p3) -- (p4) -- (p5) -- (p6);
      \draw[decoration={brace,mirror,raise=0.35cm},decorate]
        (-3,0) -- (0-.1,0) node [pos=0.5,anchor=north,yshift=-0.55cm] {$C_1$};
      \draw[decoration={brace,mirror,raise=0.35cm},decorate]
        (0+.1,0) -- (3,0) node [pos=0.5,anchor=north,yshift=-0.55cm] {$C_2$};
    }
    \begin{scope}[xshift=0*\xshift]
      \drawchain
      \begin{scope}[on background layer]
        \draw[edge] (p0) -- (p2);
        \draw[edge] (p4) -- (p6);
        \draw[edge,myred,semithick] (p2) -- (p4);
      \end{scope}
    \end{scope}
    \begin{scope}[xshift=1*\xshift]
      \drawchain
      \begin{scope}[on background layer]
        \draw[edge] (p0) -- (p2);
        \draw[edge] (p4) -- (p6);
        \draw[edge,myblue,semithick] (p2) -- (p4);
        \draw[edge,myred,semithick] (p0) -- (p4);
      \end{scope}
    \end{scope}
    \begin{scope}[xshift=2*\xshift]
      \drawchain
      \begin{scope}[on background layer]
        \draw[edge] (p0) -- (p2);
        \draw[edge] (p4) -- (p6);
        \draw[edge] (p2) -- (p4);
        \draw[edge,myblue,semithick] (p0) -- (p4);
        \draw[edge,myred,semithick] (p0) -- (p6);
      \end{scope}
    \end{scope}
  \end{tikzpicture}
  \caption{
    The three cases when removing the bridge from a partial upper
    triangulation of $C_1 \vee C_2$ in the proof of \cref{prop:convexsumdp}.
    On the left, both $C_1$ and $C_2$ gain a visible edge.
    In the middle, only $C_1$ gains a visible edge.
    On the right, only $C_2$ gains a visible edge.
    The current bridge is red, and the edge that becomes
    the new bridge is blue.
  }
  \label{fig:bridgethreecases}
\end{figure}
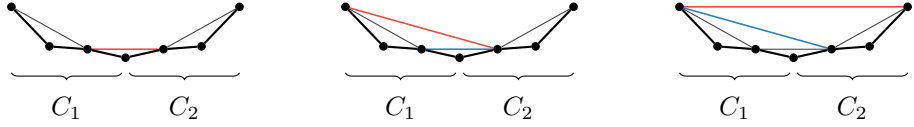

\subsection{Koch Chains}

Recall \cref{def:kochchain} and that the formula for Koch chains
expands to the nested expression
\[
K_s = (K_{s-2} \wedge K_{s-2}) \vee (K_{s-2} \wedge K_{s-2})
\]
with alternating convex and concave sums. This repeated mixing
of the two types of sums appears to make an exact analysis of the
number of triangulations of $K_s$ very difficult.

Instead, we have implemented the quadratic time algorithm from the previous subsection
and used it to compute $T_{K_s}(x)$ and $T_{\flip{K_s}}(x)$ for all $s \leq 21$. To deal with the
exponentially growing coefficients, we rely on a custom floating point type
with a $64$ bit mantissa and a $32$ bit exponent from the boost multiprecision library.
As only additions and multiplications are involved,
we do not have to deal with numerical issues; in fact, the rounding errors
grow at most linearly. In addition, we make use of multi-threading and take advantage
of symmetries of $K_s$ for a constant factor speed-up. This allows us
to compute $T_{K_{21}}(x)$ in around a day on
a regular workstation (Intel i7-6700HQ, 2.6GHz).

\Cref{tab:kochchainsnumbers} from the introduction lists the resulting numbers. For example,
$K_{21}$ has approximately $9.082799^n$ triangulations, where $n = 2^{21}$.
In the next section, we show how the computed coefficients
of $T_{K_{21}}(x)$ can be used to give bounds on $\TR(K_s)$ as $s \to \infty$.

\section{Poly Chains and Twin Chains}
\label{sec:polytwin}

Let $C_0$ be a fixed chain with $m$ chain edges. We want to define
two particular families of chains that can be built from many copies
of $C_0$ via concave and convex sums.

\begin{definition}
For integers $N \geq 1$, the \emph{poly-$C_0$ chains (of length $n = Nm$)} are the chains
\[
\Cpoly(C_0, N) = \underbrace{\flip{C_0} \vee \dots \vee \flip{C_0}}_{\text{$N$ copies}}.
\]
\end{definition}
\begin{definition}
For integers $N \geq 1$, the \emph{twin-$C_0$ chains (of length $n = 2Nm + 1$)} are the chains
\[
\Ctwin(C_0, N) = \flip{\Cpoly(C_0, N)} \vee \cOne \vee \flip{\Cpoly(C_0, N)}.
\]
\end{definition}

Note that both resulting chains are upward chains, as long as $N>1$.
For example, the poly-$\cOne$ chains are the convex chains,
the twin-$\cOne$ chains are the classic double chains,
and the twin-$(\cOne \vee \cOne)$ chains are the double zig-zag chains.

We are interested in the asymptotic behavior of the number of triangulations
of these constructions as $N$ goes to infinity. \Cref{lem:concavesum} allows us to express
the number of lower triangulations in terms of upper triangulations.
\begin{align*}
L(\Cpoly(C_0, N)) &= U(C_0 \wedge \dots \wedge C_0) = U(C_0)^N\\
L(\Ctwin(C_0, N)) &= U(\Cpoly(C_0, N)\wedge \cOne \wedge \Cpoly(C_0, N)) = U(\Cpoly(C_0, N))^2
\end{align*}

For the upper triangulations, we make use of the following general result,
whose proof can be found in the following subsections.

\begin{theorem} \label{theo:polytwinasym}
The poly-$C_0$ chains of length $n$ have $\tTheta(\lambda^{n})$ upper triangulations
(where the $\tTheta$-notation suppresses potential additional factors that are polynomial in $n$),
while the twin-$C_0$ chains of length $n$ have $\tTheta(\tau^{n})$ upper triangulations, where
\begin{align*}
\lambda = \sqrt[m]{\sum_{k=1}^{m}  2^k (k+1) \cdot t_{m-k}(\flip{C_0})},
&&
\tau = \sqrt[m]{\sum_{k=1}^{m} 2^k \cdot t_{m-k}(C_0)}.
\end{align*}
It follows that the twin-$C_0$ chains of length $n$ have $\tTheta((\lambda \tau)^{n})$ 
complete triangulations.
\end{theorem}

\begin{example}
Let us analyze the poly-$\vex{4}$ chains and twin-$\vex{4}$ chains.
We have
\begin{align*}
T_{\flip{\vex{4}}}(x) = 1, && T_{\vex{4}}(x) = 1 + 3 x + 5 x^2 + 5 x^3,
\end{align*}
which yields $\lambda = \sqrt[4]{80}$ and $\tau = \sqrt[4]{70}$.
Therefore, the twin-$\vex{4}$ chains have
$\tTheta(\sqrt[4]{5600}^{n})$ triangulations,
where $\sqrt[4]{5600} \approx 8.6506154$.
Note that these chains are the generalized double zig-zag chains from \cite{dumitrescu2013}.
By comparison, the numerical bound there was $\Omega(8.6504^n)$.
\end{example}

In the same way, by using the many coefficients of $T_{K_{s}}(x)$ and $T_{\flip{K_{s}}}(x)$
that we found for all~$s \leq 21$ with our algorithm from the previous section,
we can also compute the corresponding values of $\lambda$ and $\tau$ for poly-Koch chains
and twin-Koch chains.
The results can be found in \cref{tab:polytwinkochchain}.
The last row in particular allows us to analyze twin-$K_{21}$ chains and, therefore, prove
\cref{thm:lowerbound} from the introduction.

\begin{table}
  \centering
  \begin{tabular}{@{}rrlllrrlll@{}}
    \toprule
      $s$ & $m$ & $\lambda$ & $\tau$ & $\lambda\tau$ & $s$ & $m$ & $\lambda$ & $\tau$ & $\lambda\tau$ \\
    \cmidrule(r){1-5}\cmidrule(l){6-10}
       0 &       1 & 4.0      & 2.0      & 8.0      & 11 &    2048 & 2.894607 & 3.137597 & 9.082113 \\
       1 &       2 & 3.464101 & 2.449489 & 8.485281 & 12 &    4096 & 3.142184 & 2.890518 & 9.082542 \\
       2 &       4 & 3.534118 & 2.449489 & 8.656787 & 13 &    8192 & 2.892806 & 3.139805 & 9.082850 \\
       3 &       8 & 3.124013 & 2.841004 & 8.875335 & 14 &   16384 & 3.141127 & 2.891623 & 9.082957 \\
       4 &      16 & 3.290140 & 2.721989 & 8.955727 & 15 &   32768 & 2.892276 & 3.140445 & 9.083034 \\
       5 &      32 & 2.974654 & 3.033787 & 9.024469 & 16 &   65536 & 3.140819 & 2.891940 & 9.083061 \\
       6 &      64 & 3.191872 & 2.835019 & 9.049019 & 17 &  131072 & 2.892123 & 3.140627 & 9.083081 \\
       7 &     128 & 2.919234 & 3.106209 & 9.067752 & 18 &  262144 & 3.140731 & 2.892029 & 9.083087 \\
       8 &     256 & 3.157095 & 2.874272 & 9.074351 & 19 &  524288 & 2.892080 & 3.140677 & 9.083092 \\
       9 &     512 & 2.900536 & 3.130176 & 9.079191 & 20 & 1048576 & 3.140706 & 2.892054 & 9.083094 \\
      10 &    1024 & 3.145716 & 2.886746 & 9.080887 & 21 & 2097152 & 2.892068 & 3.140692 & 9.083095 \\
    \bottomrule
  \end{tabular}
  \caption{
    The computed values of $\lambda$ and $\tau$ as in \cref{theo:polytwinasym} for the special case $C_0 = K_s$, where each entry is rounded down to six decimal places.
  }
  \label{tab:polytwinkochchain}
\end{table}

\begin{corollary}
The twin-$K_{21}$ chains of length $n$ have $\tTheta(\tau^{n})$ upper triangulations, $\tTheta(\lambda^{n})$ lower triangulations, and hence $\tTheta((\lambda\tau)^{n})$ complete triangulations, where $\tau \approx 3.140692$, $\lambda \approx 2.892068$, and $\lambda\tau \approx 9.083095$.
\end{corollary}

The next two lemmas, when combined with the first part of \cref{theo:polytwinasym},
can further be used to prove asymptotic bounds for families of
chains that are built from the same $C_0$.

\begin{lemma} \label{lemm:copyboundupper}
Let $C$ be any chain that can be written as a formula involving
convex sums, concave sums and exactly $N$ copies of $C_0$. Then,
\[
U(C_0)^N \leq U(C) \leq U(\Cpoly(\flip{C_0}, N)).
\]
\end{lemma}
\begin{proof}
  We get $U(C_0 \wedge \dots \wedge C_0) \leq U(C) \leq U(C_0 \vee \dots \vee C_0)$ by using induction on $N$ combined with \cref{coro:vexmorecave}.
\Cref{lem:concavesum} and the definition of poly-$\flip{C_0}$ chains then concludes the proof.
\end{proof}
\begin{lemma} \label{lem:copyboundall}
In the same setting, we have
\[
\TR(C_0)^N \leq \TR(C) \leq U(\Cpoly(C_0, N)) \cdot U(\Cpoly(\flip{C_0}, N)).
\]
\end{lemma}
\begin{proof}
Apply \cref{lemm:copyboundupper} twice. First to $C$ with $C_0$,
then to $\flip{C}$ with $\flip{C_0}$.
\end{proof}

The Koch chains $K_s$ with $s \geq 21$ can be written as formulas
involving only copies of $K_{21}$ (or only copies of $\flip{K_{21}}$, depending on the parity of $s$), and so \cref{lem:copyboundall} applies to them.
For the lower bound, we obtain $9.082799^n \leq \Tr(K_s)$ for all $s \geq 21$ from the last entry in \cref{tab:kochchainsnumbers}.
For the upper bound, we make use of \cref{tab:polyflipkochchain}, which contains once more the computed values $\lambda$ for poly-$K_s$ chains,
but it also includes the corresponding values (denoted by $\overline\lambda$) for poly-$\flip{K_s}$ chains.
Note that poly-$\flip{K_s}$ chains are in fact nothing else than poly-$K_{s-1}$ chains, and so the values $\lambda$ and $\overline\lambda$
are simply shifted by one index.
For all $s\geq21$, we obtain $\Tr(K_s) \leq U(\Cpoly(K_{21},N)) \cdot U(\Cpoly(\flip{K_{21}},N)) = \tTheta((\lambda\overline\lambda)^n)$ where $\lambda\overline\lambda \approx 9.083138$.
Clearly, this proves \cref{thm:kochlimit} from the introduction.

\begin{table}
  \centering
  \def\pz{\phantom{0}}
  \begin{tabular}{@{}rrlllrrlll@{}}
    \toprule
      $s$ & $m$ & $\lambda$ & $\overline\lambda$ & $\lambda\overline\lambda$ & $s$ & $m$ & $\lambda$ & $\overline\lambda$ & $\lambda\overline\lambda$ \\
    \cmidrule(r){1-5}\cmidrule(l){6-10}
       0 &       1 & 4.0      & 4.0      &   16.0      & 11 &    2048 & 2.894607 & 3.145716 & 9.105615 \\
       1 &       2 & 3.464101 & 4.0      &   13.856406 & 12 &    4096 & 3.142184 & 2.894607 & 9.095390 \\
       2 &       4 & 3.534118 & 3.464101 &   12.242546 & 13 &    8192 & 2.892806 & 3.142184 & 9.089731 \\
       3 &       8 & 3.124013 & 3.534118 &   11.040634 & 14 &   16384 & 3.141127 & 2.892806 & 9.086674 \\
       4 &      16 & 3.290140 & 3.124013 &   10.278444 & 15 &   32768 & 2.892276 & 3.141127 & 9.085007 \\
       5 &      32 & 2.974654 & 3.290140 & \pz9.787032 & 16 &   65536 & 3.140819 & 2.892276 & 9.084117 \\
       6 &      64 & 3.191872 & 2.974654 & \pz9.494717 & 17 &  131072 & 2.892123 & 3.140819 & 9.083637 \\
       7 &     128 & 2.919234 & 3.191872 & \pz9.317823 & 18 &  262144 & 3.140731 & 2.892123 & 9.083383 \\
       8 &     256 & 3.157095 & 2.919234 & \pz9.216301 & 19 &  524288 & 2.892080 & 3.140731 & 9.083247 \\
       9 &     512 & 2.900536 & 3.157095 & \pz9.157271 & 20 & 1048576 & 3.140706 & 2.892080 & 9.083176 \\
      10 &    1024 & 3.145716 & 2.900536 & \pz9.124266 & 21 & 2097152 & 2.892068 & 3.140706 & 9.083138 \\
    \bottomrule
  \end{tabular}
  \caption{
    The computed values of $\lambda$ for poly-$K_s$ chains, the corresponding $\overline\lambda$ for poly-$\flip{K_s}$ chains,
    and the product $\lambda\overline\lambda$ for the resulting upper bound given by \cref{lem:copyboundall},
    where each entry is rounded down to six decimal places.
  }
  \label{tab:polyflipkochchain}
\end{table}

\subsection{Tools for the proof of \cref{theo:polytwinasym}}

We use similar ideas as in the analysis of the generalized double zig-zag chain by Dumitrescu
et al.\ (Section~2 of \cite{dumitrescu2013}),
with three key improvements that allow us to get an exact $\tTheta$
instead of just a numerical lower bound.

\subparagraph{First improvement.}

Our first improvement is that our chain framework allows
us to analyze even more general ``double circles'', defined as follows.

\begin{definition}\label{def:generaldoublecircle}
Let $N_1, \dots, N_m \geq 0$ be integers, at least one non-zero.
Then, the corresponding \emph{generalized double circle} (with $n = 1N_1 + 2N_2 + \dots + mN_m$ chain edges) is the chain
\[
\Cgdc(N_1, \dots, N_m) = \Cpoly(\vex{1}, N_1) \vee \Cpoly(\vex{2}, N_2) \vee \dots \vee \Cpoly(\vex{m}, N_m),
\]
where we simply omit the poly-$\vex{k}$ chains with $N_k = 0$.
\end{definition}

\begin{theorem} \label{theo:multiconvexlower}
Let $N = N_1 + \dots + N_m$.
There is a constant $c_m$ depending only on $m$
such that
\[
U(\Cgdc(N_1, \dots, N_m)) \geq \frac{1}{(N+1)^{c_m}} \prod_{k=1}^{m} \Big(2^k (k+1) \Big)^{N_k}.
\]
\end{theorem}
\begin{proof}
By \cref{coro:vexmorecave} and \cref{lem:concavesum}, we get
\[
U(\Cgdc(N_1, \dots, N_m)) \geq \prod_{k=1}^{m} U(\Cpoly(\vex{k}, N_k)).
\]
Asinowski, Krattenthaler, and Mansour~\cite{asinowski2017} show for any fixed $k$ that
\[
  U(\Cpoly(\vex{k}, M)) = \tOmega((2^k (k+1))^M) \hspace{1cm} \text{(as  $M \to \infty$)},
\]
implying the existence of an absolute constant $d_k$ such that
\[
U(\Cpoly(\vex{k}, M)) \geq \frac{1}{(M+1)^{d_k}} \cdot (2^k (k+1))^M \hspace{1cm} \text{(for all $M \geq 1$)}.
\]
Substituting the respective $N_k$ for $M$ and defining $c_m = d_1 + \dots + d_m$ concludes the proof.
\end{proof}

\subparagraph{Second improvement.}

The authors of \cite{dumitrescu2013} use numerical optimization in the final
step of analyzing generalized double zig-zag chains.
We solve the corresponding optimization problem algebraically,
via Lagrange multipliers. This is best captured by the following lemma.

\begin{lemma} \label{lemm:entropymaximum}
Let $u_1, \dots, u_m \geq 0$ be given and
let $H(\alpha_1, \dots, \alpha_m) = - \sum_{k=1}^m \alpha_k \cdot \ln (\alpha_k)$
be the entropy function. Then,
\[
\max_{\substack{0 \leq \alpha_1, \dots, \alpha_m \leq 1\\ \alpha_1 + \dots + \alpha_m = 1}} e^{H(\alpha_1, \dots, \alpha_m)} \cdot \prod_{k=1}^{m} u_k^{\alpha_k} = \sum_{k=1}^{m} u_k.
\]
\end{lemma}
\begin{proof}
Without loss of generality, assume that $u_k > 0$ for all $k$ (otherwise put $\alpha_k = 0$, after which we are left with the same proof for a smaller number $m$ of variables). To simplify notation, define
\[
f(\alpha_1, \dots, \alpha_m) := e^{H(\alpha_1, \dots, \alpha_m)} \cdot \prod_{k=1}^{m} u_k^{\alpha_k}.
\]
The proof is by induction on $m$.
If $m = 1$, then $\alpha_1 = 1$ and the statement is trivial.
For $m \geq 2$, note that the maximum of $f$ always exists since we are maximizing over a compact region.

Let us now consider the boundary first.
If $\alpha_i = 0$ for some index $i$, then by induction
\[
\max_{\substack{0 \leq \alpha_1, \dots, \alpha_m \leq 1\\ \alpha_1 + \dots + \alpha_m = 1\\ \alpha_i = 0}} f(\alpha_1, \dots, \alpha_m) = \sum_{k=1}^{m} u_k - u_i \leq \sum_{k=1}^{m} u_k.
\]
On the other hand, if $\alpha_i = 1$ for some $i$, then $\alpha_{j} = 0$ for all other indices $j$ and the same argument applies. To summarize, if the maximum is on the boundary, then it is at most $u_1 + \dots + u_k$.

Let us now consider the interior.
After taking logs, we are left to show that
\[
\max_{\substack{0 < \alpha_1, \dots, \alpha_m < 1\\ \alpha_1 + \dots + \alpha_m = 1}} \ln(f(\alpha_1, \dots, \alpha_m)) = \ln \Big(\sum_{k=1}^{m} u_k\Big).
\]
By making use of the definitions
\begin{align*}
g(\alpha_1, \dots, \alpha_m) &:= \ln(f(\alpha_1, \dots, \alpha_m)) =  H(\alpha_1, \dots, \alpha_m) + \sum_{k=1}^{m} \alpha_k \cdot \ln(u_k),\\
h(\alpha_1, \dots, \alpha_m) &:= \alpha_1 + \dots + \alpha_m - 1,
\end{align*}
we see that, equivalently, we are maximizing $g$ over the open set $0 < \alpha_1, \dots, \alpha_m < 1$
under the constraint $h(\alpha_1, \dots, \alpha_m) = 0$. By Lagrange multipliers,
any such maximum $(\alpha_1, \dots, \alpha_m)$, if it exists, has to solve
\[
\frac{\partial g}{\partial \alpha_k} = \lambda \cdot \frac{\partial h }{\partial \alpha_k}
\]
for some $\lambda \in \bR$ and all $k = 1,\dots,m$. For our concrete functions $g$ and $h$, these equations are
\[
-\ln(\alpha_k) - 1 + \ln(u_k) = \lambda
\]
or, equivalently,
\[
\alpha_k = e^{-\lambda - 1} \cdot u_k.
\]
Plugging this into the constraint
$\alpha_1 + \dots + \alpha_m = 1$ yields
\[
e^{-\lambda - 1} \cdot (u_1 + \dots + u_m) = 1,
\]
which then allows us to eliminate $\lambda$ and we arrive at
\[
\alpha_k = \frac{1}{u_1 + \dots + u_m} \cdot u_k.
\]
For these particular values $\alpha_k$, we now calculate
\begin{align*}
g(\alpha_1, \dots, \alpha_m) &= H(\alpha_1, \dots, \alpha_m) + \sum_{k=1}^{m} \alpha_k \cdot \ln(u_k)\\
&= - \sum_{k=1}^{m} \frac{u_k}{u_1 + \dots + u_m} \cdot \ln \Big(\frac{u_k}{u_1 + \dots + u_m}\Big) + \sum_{k=1}^{m} \frac{u_k}{u_1 + \dots + u_m} \cdot \ln(u_k)\\
&= \sum_{k=1}^{m} \frac{u_k}{u_1 + \dots + u_m} \cdot \Big( - \ln \Big(\frac{u_k}{u_1 + \dots + u_m}\Big) + \ln (u_k) \Big)\\
&= \sum_{k=1}^{m} \frac{u_k}{u_1 + \dots + u_m} \cdot \ln(u_1 + \dots + u_m) \\
&= \ln(u_1 + \dots + u_m)
\end{align*}
and, hence, $f(\alpha_1, \dots, \alpha_m) = u_1 + \dots + u_m$.
Given that this value is at least as large as what we got for the boundary,
there has to be a maximum in the interior. By Lagrange multipliers,
that maximum has to be at the values $\alpha_k$ we just considered.
\end{proof}

\subparagraph{Third improvement.}

Finally, we define a special generating
function that behaves well with regards to convex sums.
This will allow us to prove a matching upper bound for
\cref{theo:multiconvexlower}.

\begin{definition}
Let $C$ be any chain with $n$ chain edges. The \emph{upper triangulation generating function of $C$} is defined as
\[
\phi_C(x) = T_C(x) - \Big(\frac{x}{1-x}\Big)^{n+1} T_C(1-x).
\]
\end{definition}
Note that $\phi_C(x)$ is a rational function. As a formal power series,
it satisfies $\phi_C(x) = T_C(x) + O(x^{n+1})$. In other words, for $k = 0,\dots,n$,
the coefficient of $x^{k}$ in the Taylor expansion of $\phi_C(x)$ at $x = 0$
is equal to the triangulation number $t_k(C)$.

\begin{lemma} \label{lem:tgfvee}
For any two chains $C_1$ and $C_2$, we have
\[
\phi_{C_1 \vee C_2}(x) = \frac{1-x}{1-2x} \cdot \phi_{C_1}(x) \cdot \phi_{C_2}(x).
\]
\end{lemma}

The proof of \cref{lem:tgfvee} is deferred to the next subsection since it involves some tedious computations.
Here, we focus on the following implication.

\begin{theorem} \label{thm:multiconvexupper}
For the generalized double circle, we have
\[
U(\Cgdc(N_1, \dots, N_m)) \leq \prod_{k=1}^{m} (2^k (k+1) )^{N_k}.
\]
\end{theorem}
\begin{proof}
First, by repeated invocation of \cref{lem:tgfvee}, we get
\[
\phi_{\Cgdc(N_1, \dots, N_m)}(x) = \Big(\frac{1-x}{1-2x}\Big)^{N_1 + \dots + N_m-1} \cdot  \prod_{k=1}^{m} \Big(\phi_{\cave{k}}(x)\Big)^{N_k}.
\]

Second, recall that for concave chains we have $T_{\cave{k}}(x) = 1$.
Therefore, by definition of $\phi$ and by the geometric series,
\[
\phi_{\cave{k}}(x) = 1 - \Big(\frac{x}{1-x}\Big)^{k+1} = \frac{1-2x}{1-x} \cdot \sum_{i=0}^{k} \Big(\frac{x}{1-x}\Big)^i.
\]

Third, for all $n \geq 2$ and $k \geq 0$ we have the combinatorial identity
\[
[x^{n-1}] \Big(\frac{x}{1-x}\Big)^k = \binom{n-2}{k-1} < 2^n.
\]

With $n = 1N_1 + 2N_2 + \dots + mN_m$ being the number of chain edges, we can now combine these three ingredients and get
\begin{align*}
  U(\Cgdc(N_1, \dots, N_m)) & = [x^{n-1}] \phi_{\Cgdc(N_1, \dots, N_m)}(x) \\
& = [x^{n-1}] \Big(\frac{1-x}{1-2x}\Big)^{N_1 + \dots + N_m-1} \cdot  \prod_{k=1}^{m} \Big(\phi_{\cave{k}}(x)\Big)^{N_k} \\
& = [x^{n-1}] \frac{1-2x}{1-x} \prod_{k=1}^{m} \Bigg(\sum_{i=0}^{k} \Big(\frac{x}{1-x}\Big)^i\Bigg)^{\;\mathclap{N_k}} \\
& \leq[x^{n-1}] \prod_{k=1}^{m} \Bigg(\sum_{i=0}^{k} \Big(\frac{x}{1-x}\Big)^i\Bigg)^{\;\mathclap{N_k}} \leq 2^{n} \prod_{k=1}^{m} (k + 1)^{N_k},
\end{align*}
where expanding the second to last term yields $\prod_k (k + 1)^{N_k}$ summands,
each some power of $\frac{x}{1-x}$, for which the coefficient of  $x^{n-1}$ is always less than $2^{n}$.
\end{proof}

\subsection{Proof of \cref{lem:tgfvee}}

First, let us consider the simpler case where $C_1 = \cave{n_1}$ and $C_2 = \cave{n_2}$ for arbitrary integers $n_1,n_2 \geq 1$.
As already seen in the previous subsection, in such a case we have
\[
\phi_{\cave{n_i}}(x) = 1 - \Big(\frac{x}{1-x}\Big)^{n_i+1}.
\]
In the paragraph after \cref{lem:convexsum}, we have further seen that
\[
T_{\cave{n_1} \, \vee\, \cave{n_2}}(x) = 1 + \sum_{l=1}^{n_1} \sum_{r=1}^{n_2} \binom{l+r-2}{l-1} x^{l+r-1}.
\]
Therefore, in order to prove \cref{lem:tgfvee} for this specific choice of $C_1$ and $C_2$,
it suffices to show the following identity.

\begin{proposition}  \label{prop:verylongcomputation}
Let $n_1, n_2 \geq 0$. Define the formal power series
\begin{align*}
T_{n_1, n_2}(x) &:= 1 + \sum_{l=1}^{n_1} \sum_{r=1}^{n_2} \binom{l+r-2}{l-1} x^{l+r-1}
= 1 + \sum_{l=0}^{n_1-1} \sum_{r=0}^{n_2-1} \binom{l+r}{l} x^{l+r+1}, \\
L_{n_1, n_2}(x) &:= T_{n_1, n_2}(x) - \ppa{\frac{x}{1-x}}^{n_1+n_2+1} T_{n_1, n_2}(1-x), \\
R_{n_1, n_2}(x) &:= \frac{1-x}{1-2x} \cdot \ppa{1 - \ppa{\frac{x}{1-x}}^{n_1+1}} \cdot \ppa{1 - \ppa{\frac{x}{1-x}}^{n_2+1}}.
\end{align*}
Then, we have that $L_{n_1, n_2}(x) = R_{n_1, n_2}(x)$.
\end{proposition}
\begin{proof}
We use induction on $(n_1, n_2)$.
If at least one of $n_1$ and $n_2$ is zero, then
the proposition follows easily.
We may hence assume that $n_1, n_2 \geq 1$.

First, we do some computations with $T_{n_1, n_2}(x)$.
We use the letter $k$ to denote the sum $l+r$, and we also use Iverson brackets to simplify notation.
\begin{align*}
T_{n_1, n_2}(x)
& = 1 + \sum_{k=0}^{\infty} \sum_{l+r=k} \bra{0 \leq l < n_1, 0 \leq r < n_2} \binom{k}{l} x^{k+1} \\
& = 1 + \sum_{k=0}^{\infty} 2^k x^{k+1} - \sum_{k=n_1}^{\infty} \sum_{l=n_1}^{k} \binom{k}{l} x^{k+1} - \sum_{k=n_2}^{\infty} \sum_{r=n_2}^{k} \binom{k}{r} x^{k+1} + \sum_{l=n_1}^{\infty} \sum_{r=n_2}^{\infty} \binom{l+r}{l} x^{l+r+1}\\
& = 1 + \frac{x}{1-2x} - \sum_{l=n_1}^{\infty} \sum_{k=l}^{\infty} \binom{l+(k-l)}{l} x^{k+1} - \sum_{r=n_2}^{\infty} \sum_{k=r}^{\infty} \binom{r+(k-r)}{r} x^{k+1} + \sum_l \sum_r (\dots)\\
& = \frac{1-x}{1-2x} - \sum_{l=n_1}^{\infty} \ppa{\frac{x}{1-x}}^{l+1} - \sum_{r=n_2}^{\infty} \ppa{\frac{x}{1-x}}^{r+1} + \sum_l \sum_r (\dots)\\
& = \frac{1-x}{1-2x} - \ppa{\frac{x}{1-x}}^{n_1+1} \frac{1}{1-\frac{x}{1-x}} - \ppa{\frac{x}{1-x}}^{n_2+1} \frac{1}{1-\frac{x}{1-x}} + \sum_l \sum_r (\dots)\\
& = \frac{1-x}{1-2x} \ppa{1 - \ppa{\frac{x}{1-x}}^{n_1+1} - \ppa{\frac{x}{1-x}}^{n_2+1}} + \sum_{l=n_1}^{\infty} \sum_{r=n_2}^{\infty} \binom{l+r}{l} x^{l+r+1}
\end{align*}
We can now use the above equation in order to compute the following expression.
\begin{align*}
& T_{n_1, n_2}(x) - \frac{x}{1-x} \ppa{T_{n_1-1, n_2}(x) + T_{n_1, n_2-1}(x)} + \ppa{\frac{x}{1-x}}^2 T_{n_1-1, n_2-1}(x)\\
= \; & \frac{1-2x}{1-x} + \frac{1}{(1-x)^2} \sum_{l=n_1}^{\infty} \sum_{r=n_2}^{\infty} \bigg[(1-x)^2 \binom{l+r}{l} - (1-x) \binom{l+r-1}{l}\\
& \hspace{4cm} - (1-x) \binom{l+r-1}{l-1} + \binom{l+r-2}{l-1}\bigg] \cdot x^{l+r+1}\\
= \; & \frac{1-2x}{1-x} + \frac{1}{(1-x)^2} \sum_{l=n_1}^{\infty} \sum_{r=n_2}^{\infty} \bigg[x^2 \binom{l+r}{l} - x \binom{l+r-1}{l}\\
& \hspace{4cm}  - x \binom{l+r-1}{l-1} + \binom{l+r-2}{l-1}\bigg] \cdot x^{l+r+1}\\
= \; & \frac{1-2x}{1-x} + \frac{1}{(1-x)^2} \sum_{l \in \bZ} \sum_{r \in \bZ} \binom{l+r-2}{l-1} \bigg(\bra{l > n_1, r > n_2}  - \bra{l > n_1, r \geq n_2}\\
& \hspace{4cm} - \bra{l \geq n_1, r > n_2} + \bra{l \geq n_1, r \geq n_2}\bigg) \cdot x^{l+r+1}\\
= \; & \frac{1-2x}{1-x} + \frac{1}{(1-x)^2} \sum_{l \in \bZ} \sum_{r \in \bZ} \binom{l+r-2}{l-1} \bra{l=n_1, r=n_2} \cdot x^{l+r+1}\\
= \; & \frac{1-2x}{1-x} + \frac{1}{(1-x)^2} \binom{n_1+n_2-2}{n_1-1} x^{n_1+n_2+1}\\
\end{align*}
In turn, we can further use the above to compute a very similar expression for $L_{n_1,n_2}(x)$,
noting that the second term from before cancels nicely in the very last step.
\begin{align*}
& L_{n_1, n_2}(x) - \frac{x}{1-x} \ppa{L_{n_1-1, n_2}(x) + L_{n_1, n_2-1}(x)} + \ppa{\frac{x}{1-x}}^2 L_{n_1-1, n_2-1}(x)\\
= \; & T_{n_1, n_2}(x) - \frac{x}{1-x} \ppa{T_{n_1-1, n_2}(x) + T_{n_1, n_2-1}(x)} + \ppa{\frac{x}{1-x}}^2 T_{n_1-1, n_2-1}(x)\\
& \quad - \ppa{\frac{x}{1-x}}^{n_1+n_2+1} \binom{n_1+n_2-2}{n_1-1} (1-x)^{n_1+n_2-1}
= \frac{1-2x}{1-x}
\end{align*}
On the other hand, a direct computation shows that we also have the same equation for $R_{n_1,n_2}(x)$.
\[
R_{n_1, n_2}(x) - \frac{x}{1-x} \ppa{R_{n_1-1, n_2}(x) + R_{n_1, n_2-1}(x)} + \ppa{\frac{x}{1-x}}^2 R_{n_1-1, n_2-1}(x) = \frac{1-2x}{1-x}
\]
Therefore, the proposition now easily follows by combining the previous two equations with the induction hypothesis.
\end{proof}

Let now $C_1$ and $C_2$ be two arbitrary chains
with $n_1$ and $n_2$ chain edges, respectively.
By making use of \cref{lem:convexsum} and the special case of \cref{lem:tgfvee} that
we just proved in \cref{prop:verylongcomputation}, we can now compute
\begin{align*}
  \phi_{C_1 \vee C_2}(x)
 &= T_{C_1 \vee C_2}(x) + \Big(\frac{x}{1-x}\Big)^{n_1+n_2+1} T_{C_1 \vee C_2}(1-x)\\
&= \sum_{k_1=0}^{n_1-1} \sum_{k_2=0}^{n_2-1} t_{k_1}(C_1) \cdot t_{k_2}(C_2) \cdot x^{k_1+k_2} \Bigg(T_{\cave{n_1 - k_1} \, \vee\, \cave{n_2 - k_2}}(x)\\
&\hspace{2cm}+ \Big(\frac{x}{1-x}\Big)^{n_1+n_2+1-k_1-k_2} \cdot T_{\cave{n_1 - k_1} \, \vee\, \cave{n_2 - k_2}}(1-x)\Bigg)\\
&= \sum_{k_1=0}^{n_1-1} \sum_{k_2=0}^{n_2-1} t_{k_1}(C_1) \cdot t_{k_2}(C_2) \cdot x^{k_1+k_2} \cdot \phi_{\cave{n_1 - k_1} \, \vee\, \cave{n_2 - k_2}}(x)\\
&= \sum_{k_1=0}^{n_1-1} \sum_{k_2=0}^{n_2-1} t_{k_1}(C_1) \cdot t_{k_2}(C_2) \cdot x^{k_1+k_2} \cdot \frac{1-x}{1-2x} \cdot \phi_{\cave{n_1 - k_1}}(x) \cdot \phi_{\cave{n_2 - k_2}}(x)\\
&= \frac{1-x}{1-2x} \cdot \Big(\sum_{k_1=0}^{n_1-1} t_{k_1}(C_1) \cdot x^{k_1} \cdot \phi_{\cave{n_1 - k_1}}(x) \Big) \cdot \Big(\sum_{k_2=0}^{n_2-1} t_{k_2}(C_2) \cdot x^{k_2} \cdot \phi_{\cave{n_2 - k_2}}(x) \Big)\\
&= \frac{1-x}{1-2x} \cdot \phi_{C_1}(x) \cdot \phi_{C_2}(x).
\end{align*}

\subsection{Proof of \cref{theo:polytwinasym} (Poly Chains)}

Recall once more \cref{lem:convexsum}.
For the special case where $C_2 = \cave{n_2}$, we have that $t_0(C_2) = 1$ is the only non-zero triangulation number,
and hence
\[
T_{C_1 \vee \cave{n_2}}(x) = \sum_{k_1=0}^{n_1-1} t_{k_1}(C_1) \cdot x^{k_1} \cdot T_{\cave{n_1 - k_1} \, \vee\, \cave{n_2}}(x).
\]
After plugging this into the right-hand side of \cref{lem:convexsum},
we get the following variant of the statement, where the second sum is left unexpanded.

\begin{lemma} \label{lemm:onesidedsum}
Let $C_1$ be a chain with $n_1$ chain edges and let $C_2$ be any chain. Then,
\[
T_{C_1 \vee C_2}(x) = \sum_{k=0}^{n_1-1} t_{k}(C_1) \cdot x^{k} \cdot T_{\cave{n_1 - k} \, \vee\, C_2}(x).
\]
\end{lemma}

Let now $C_0$ be any fixed chain with $m$ chain edges and let $N \geq 2$.
Using commutativity as in \cref{coro:triangcommutative},
we can inductively apply \cref{lemm:onesidedsum} to the poly-$C_0$ chain
and get
\[
T_{\Cpoly(C_0, N)} = \sum_{k_1=0}^{m-1} \dots \sum_{k_N=0}^{m-1} t_{k_1}(\flip{C_0}) \cdots t_{k_N}(\flip{C_0}) \cdot x^{k_1 + \dots + k_N} \cdot T_{\cave{m - k_1} \, \vee\, \dots \vee\, \cave{m - k_N}}(x).
\]
For any particular summand in the above expression,
let $a_i$ be the number of indices $j = 1,\dots,N$ that satisfy $k_j = m-i$.
Then, this particular summand can also be written as
\[
t_{0}(\flip{C_0})^{a_m} \cdot t_{1}(\flip{C_0})^{a_{m-1}} \cdots t_{m-1}(\flip{C_0})^{a_1}\cdot x^{0a_{m} + 1a_{m-1} + \dots + (m-1) a_1} \cdot T_{\Cgdc(a_1, \dots, a_m)}(x).
\]
Since the number of summands with these particular values $a_i$ is equal to
$\binom{N}{a_1, \dots, a_m}$, we hence get
\[
T_{\Cpoly(C_0, N)}(x) = \sum_{\substack{0 \leq a_1, \dots, a_m \leq N\\ a_1 + \dots + a_m = N}} \binom{N}{a_1, \dots, a_m} \cdot \Bigg( \prod_{k=1}^{m} t_{m-k}(\flip{C_0})^{a_k} \cdot x^{(m-k)a_k} \Bigg) \cdot T_{\Cgdc(a_1, \dots, a_m)}(x).
\]
Finally, note that the corresponding leading coefficients on both sides are equal to
\[ 
U(\Cpoly(C_0, N)) = \sum_{\substack{0 \leq a_1, \dots, a_m \leq N\\ a_1 + \dots + a_m = N}} \binom{N}{a_1, \dots, a_m} \cdot \Bigg( \prod_{k=1}^{m} t_{m-k}(\flip{C_0})^{a_k} \Bigg) \cdot U(\Cgdc(a_1, \dots, a_m)).
\]
We will now use this equation to prove both an upper and a lower bound for $U(\Cpoly(C_0, N))$.

\subparagraph{Upper bound.} By \cref{thm:multiconvexupper} and the multinomial theorem, we get
\begin{align*}
U(\Cpoly(C_0, N))
&\leq \sum_{\substack{0 \leq a_1, \dots, a_m \leq N\\ a_1 + \dots + a_m = N}} \binom{N}{a_1, \dots, a_m} \cdot \prod_{k=1}^{m} \Big(t_{m-k}(\flip{C_0}) \cdot 2^k (k+1)\Big)^{a_k}\\
&= \Bigg(\sum_{k=1}^{m} t_{m-k}(\flip{C_0}) \cdot 2^k (k+1) \Bigg)^N.
\end{align*}
The claimed upper bound now follows by recalling that the poly-$C_0$ chain has $n = Nm$ chain edges.

\subparagraph{Lower bound.} The entropy bound for multinomial coefficients states that
\[
\binom{N}{a_1, \dots, a_m} \geq \frac{1}{N^m} e^{H(\frac{a_1}{N}, \dots, \frac{a_m}{N})}, \qquad H(\alpha_1, \dots,\alpha_m) = -\sum_{k=1}^{m} \alpha_k \ln(\alpha_k).
\]
By \cref{theo:multiconvexlower} and the above entropy bound, we get for some constant $c_m$ that
\begin{align*}
U(\Cpoly(C_0, N))
&\geq \frac{1}{(N+1)^{c_m}} \sum_{\substack{0 \leq a_1, \dots, a_m \leq N\\ a_1 + \dots + a_m = N}} \binom{N}{a_1, \dots, a_m} \cdot \prod_{k=1}^{m} \Big(t_{m-k}(\flip{C_0}) \cdot 2^k (k+1)\Big)^{a_k}\\
&\geq \frac{1}{(N+1)^{c_m+m}}\sum_{\substack{0 \leq a_1, \dots, a_m \leq N\\ a_1 + \dots + a_m = N}} e^{H(\frac{a_1}{N}, \dots, \frac{a_m}{N})} \cdot \prod_{k=1}^{m} \Big(t_{m-k}(\flip{C_0}) \cdot 2^k (k+1)\Big)^{a_k}.
\end{align*}
We can lower bound the right hand side by picking any one summand.
We pick the largest one by using \cref{lemm:entropymaximum}
(we might have to round the $a_k$ slightly to get integer values,
but the effect of this vanishes as $N$ goes to infinity).
As desired, the resulting lower bound is therefore
\[
U(\Cpoly(C_0, N)) = \Omega\Bigg(\frac{1}{(N+1)^{c_m+m}} \Bigg( \sum_{k=1}^{m} t_{m-k}(\flip{C_0}) \cdot 2^k (k+1)\Bigg)^N\Bigg) \hspace{1cm} \text{(as $N \to \infty$)}.
\]

\subsection{Proof of \cref{theo:polytwinasym} (Twin Chains)}
Let $C_0$ again be any fixed chain with $m$ chain edges and let $N \geq 2$.
Observe then with \cref{lem:concavesum} that we have in fact $T_{\flip{\Cpoly(C_0, N)}}(x) = T_{C_0}(x)^N$ for the flipped poly-$C_0$ chains.
Therefore, expanding the right hand side with the multinomial theorem yields the coefficients
\[
\gamma_{k,N} := t_k(\flip{\Cpoly(C_0, N)}) = \sum_{\substack{0 \leq a_1, \dots, a_m \leq N\\a_1 + \dots + a_m = N\\(m-1) a_1 + \dots + a_{m-1} = k}} \binom{N}{a_1, \dots, a_m} \prod_{\ell=1}^{m} t_{m-\ell}(C_0)^{a_\ell}.
\]
On the other hand, applying \cref{lemm:onesidedsum} twice to the twin-$C_0$ chains yields
\[
T_{\Ctwin(C_0, N)} = \sum_{k_1=0}^{m N-1} \sum_{k_2=0}^{m N-1} \gamma_{k_1,N} \cdot \gamma_{k_2,N} \cdot x^{k_1 + k_2} \cdot T_{\cave{m N-k_1} \vee E \vee \cave{m N - k_2}}(x),
\]
with leading coefficients
\begin{align*}
U({\Ctwin(C_0, N)})
& = \sum_{k_1=0}^{m N-1} \sum_{k_2=0}^{m N-1} \gamma_{k_1,N} \cdot \gamma_{k_2,N} \cdot U({\cave{m N-k_1} \vee E \vee \cave{m N - k_2}}) \\
& = \sum_{k_1=0}^{m N-1} \sum_{k_2=0}^{m N-1} \gamma_{k_1,N} \cdot \gamma_{k_2,N} \cdot \binom{2 m N - k_1 - k_2}{m N - k_1},
\end{align*}
where the last step is the same as in the analysis of the classic double chain, see~\cite{garcia2000}.
We will again use the above equation to prove both an upper and a lower bound for $U({\Ctwin(C_0, N)})$.

\subparagraph{Upper bound.} After bounding the binomial coefficient by $2^{2 m N - k_1 - k_2}$,
we can factorize to
\[
U({\Ctwin(C_0, N)}) \leq \Bigg(\sum_{k=0}^{m N-1} \gamma_{k,N} \cdot 2^{m N - k} \Bigg)^2.
\]
Plugging in the definition of $\gamma_{k,N}$ and noting that $m N - k = 1 a_1 + 2a_2 + \dots + ma_m$,
we can again simplify with the multinomial theorem,
similar to what we did for poly-$C_0$ chains, and get
\[
U({\Ctwin(C_0, N)}) \leq \Bigg(\sum_{k=1}^{m} t_{m-k}(C_0) \cdot 2^{k}\Bigg)^{2 N}.
\]

\subparagraph{Lower bound.} By only keeping terms with $k_1 = k_2$, we get
\begin{align*}
U({\Ctwin(C_0, N)})
&\geq \sum_{k=0}^{m N - 1} \gamma_{k,N}^2 \cdot \binom{2 m N - 2 k}{m N - k} \\
&\geq \sum_{k=0}^{m N-1} \gamma_{k,N}^2 \cdot \frac{2^{2 m N - 2 k}}{2 m N}
= \frac{4^{m N}}{2 m N} \sum_{k=0}^{m N-1} \Big(2^{-k} \cdot \gamma_{k,N}\Big)^2.
\end{align*}
By dropping all the summands except for one specific $k$ to be determined later, and by plugging in the definition of $\gamma_{k,N}$ and again noting that $m N - k = 1 a_1 + 2a_2 + \dots + ma_m$, we get
\begin{align*}
U({\Ctwin(C_0, N)}) &\geq \frac{4^{m N}}{2 m N} \Big(2^{-k} \cdot \gamma_{k,N}\Big)^2 = \frac{1}{2 m N} \Big(2^{mN-k} \cdot \gamma_{k,N}\Big)^2\\
&= \frac{1}{2 m N} \Bigg(\sum_{\substack{0 \leq a_1, \dots, a_m \leq N\\a_1 + \dots + a_m = N\\(m-1) a_1 + \dots + a_{m-1} = k}} \binom{N}{a_1, \dots, a_m} \prod_{\ell=1}^{m} \Big(t_{m-\ell}(C_0) \cdot 2^{\ell}\Big)^{a_\ell}\Bigg)^2.
\end{align*}
After applying the entropy bound,
we pick the largest summand for the best possible $k$ via Lemma~\ref{lemm:entropymaximum}, and get
\begin{align*}
U({\Ctwin(C_0, N)}) = \Omega \Bigg(\frac{1}{2 m N^{m+1}} \Bigg(\sum_{\ell=1}^{m} t_{m-\ell}(C_0) \cdot 2^{\ell}\Bigg)^{2 N}\Bigg) \hspace{1cm} \text{(as $N \to \infty$)}.
\end{align*}

\subsection*{Acknowledgments}
The material presented in this paper originates from the first author's Master's thesis \cite{rutschmann2021} under the second author's direct supervision. Both authors wish to express their gratitude to Emo Welzl, the official advisor in this endeavor.

\bibliography{references}

\begin{thebibliography}{10}

\bibitem{oeiscatalan}
OEIS Foundation~Inc. (2021).
\newblock The on-line encyclopedia of integer sequences.
\newblock \url{https://oeis.org/A000108}.
\newblock Catalan numbers.

\bibitem{oeisbigschroeder}
OEIS Foundation~Inc. (2021).
\newblock The on-line encyclopedia of integer sequences.
\newblock \url{https://oeis.org/A006318}.
\newblock Large Schr\"oder numbers.

\bibitem{oeislittleschroeder}
OEIS Foundation~Inc. (2021).
\newblock The on-line encyclopedia of integer sequences.
\newblock \url{https://oeis.org/A001003}.
\newblock Little Schr\"oder numbers.

\bibitem{aichholzer2016}
Oswin Aichholzer, Victor Alvarez, Thomas Hackl, Alexander Pilz, Bettina
  Speckmann, and Birgit Vogtenhuber.
\newblock An improved lower bound on the minimum number of triangulations.
\newblock In {\em Proceedings of the 32nd International Symposium on
  Computational Geometry}, 2016.

\bibitem{aichholzer2007}
Oswin Aichholzer, Thomas Hackl, Clemens Huemer, Ferran Hurtado, Hannes Krasser,
  and Birgit Vogtenhuber.
\newblock On the number of plane geometric graphs.
\newblock {\em Graphs Comb.}, 23(Supplement-1):67--84, 2007.

\bibitem{ajtai1982}
Mikl\'os Ajtai, V\'aclav Chv\'atal, Monroe~M. Newborn, and Endre Szemer\'edi.
\newblock Crossing-free subgraphs.
\newblock In {\em Theory and Practice of Combinatorics}, volume~60 of {\em
  North-Holland Mathematics Studies}, pages 9--12. North-Holland, 1982.

\bibitem{alvarez2013}
Victor Alvarez and Raimund Seidel.
\newblock A simple aggregative algorithm for counting triangulations of planar
  point sets and related problems.
\newblock In {\em Proceedings of the 29th the Symposium on Computational
  Geometry}, 2013.

\bibitem{asinowski2017}
Andrei Asinowski, Christian Krattenthaler, and Toufik Mansour.
\newblock Counting triangulations of some classes of subdivided convex
  polygons.
\newblock {\em Eur. J. Comb.}, 62:92--114, 2017.

\bibitem{avis1996}
David Avis and Komei Fukuda.
\newblock Reverse search for enumeration.
\newblock {\em Discret. Appl. Math.}, 65(1-3):21--46, 1996.

\bibitem{bespamyatnikh2002}
Sergei Bespamyatnikh.
\newblock An efficient algorithm for enumeration of triangulations.
\newblock {\em Comput. Geom.}, 23(3):271--279, 2002.

\bibitem{denny1997}
Markus Denny and Christian Sohler.
\newblock Encoding a triangulation as a permutation of its point set.
\newblock In {\em Proceedings of the 9th Canadian Conference on Computational
  Geometry}, 1997.

\bibitem{dumitrescu2013}
Adrian Dumitrescu, Andr{\'{e}} Schulz, Adam Sheffer, and Csaba~D. T{\'{o}}th.
\newblock Bounds on the maximum multiplicity of some common geometric graphs.
\newblock {\em {SIAM} J. Discret. Math.}, 27(2):802--826, 2013.

\bibitem{epstein1994}
Peter Epstein and J{\"{o}}rg{-}R{\"{u}}diger Sack.
\newblock Generating triangulations at random.
\newblock {\em {ACM} Trans. Model. Comput. Simul.}, 4(3):267--278, 1994.

\bibitem{goodman1983}
Jacob~E. Goodman and Richard Pollack.
\newblock Multidimensional sorting.
\newblock {\em {SIAM} J. Comput.}, 12(3):484--507, 1983.

\bibitem{hurtado1997}
Ferran Hurtado and Marc Noy.
\newblock Counting triangulations of almost-convex polygons.
\newblock {\em Ars Comb.}, 45:169--179, 1997.

\bibitem{marx2016}
D{\'{a}}niel Marx and Tillmann Miltzow.
\newblock Peeling and nibbling the cactus: Subexponential-time algorithms for
  counting triangulations and related problems.
\newblock In {\em Proceedings of the 32nd International Symposium on
  Computational Geometry}, 2016.

\bibitem{garcia2000}
Alfredo~Garc{\'{\i}}a Olaverri, Marc Noy, and Javier Tejel.
\newblock Lower bounds on the number of crossing-free subgraphs of
  k\({}_{\mbox{n}}\).
\newblock {\em Comput. Geom.}, 16(4):211--221, 2000.

\bibitem{rutschmann2021}
Daniel Rutschmann.
\newblock On chains and point configurations with many triangulations.
\newblock Master's thesis, ETH Zurich, Z\"urich, Switzerland, 2021.

\bibitem{santos2003}
Francisco Santos and Raimund Seidel.
\newblock A better upper bound on the number of triangulations of a planar
  point set.
\newblock {\em J. Comb. Theory, Ser. {A}}, 102(1):186--193, 2003.

\bibitem{seidel1998}
Raimund Seidel.
\newblock On the number of triangulations of planar point sets.
\newblock {\em Comb.}, 18(2):297--299, 1998.

\bibitem{sharir2011}
Micha Sharir and Adam Sheffer.
\newblock Counting triangulations of planar point sets.
\newblock {\em Electron. J. Comb.}, 18(1):1--74, 2011.

\bibitem{sharir2006}
Micha Sharir and Emo Welzl.
\newblock Random triangulations of planar point sets.
\newblock In {\em Proceedings of the 22nd {ACM} Symposium on Computational
  Geometry}, 2006.

\bibitem{smith1989}
Warren~D. Smith.
\newblock {\em Studies in Computational Geometry Motivated by Mesh Generation}.
\newblock PhD thesis, Princeton University, Princeton, USA, 1989.

\end{thebibliography}

\end{document}